\documentclass{article}

\usepackage{amssymb,mathrsfs,xcolor,bbold,comment,enumerate,graphicx,float}
\usepackage{amsthm}
\usepackage{amsmath}

\usepackage{tikz}
\usepackage{tikz-cd}
\usepackage{enumerate}
\usepackage[left=3cm, right=3cm]{geometry}
\usepackage{hyperref}
\usepackage{cancel}

\newcommand{\PWp}{{\mathbb P}}
\newcommand{\PWq}{{\mathbb Q}}

\newcommand{\abs}[1]{\lvert#1\rvert}
\newcommand{\norm}[1]{\lVert#1\rVert}

\newcommand{\R}{\mathbb{R}} 

\newcommand{\QW}{\mathbb{Q}}

\newcommand{\E}{\mathbb{E}}

\newcommand{\ind}{\mathbf{1}}

\newcommand{\Space}{(\Omega,\mathcal{F},\mathbb{P})}

\newcommand{\PW}{\mathbb{P}}

\newcommand{\N}{\mathbb{N}}

\newcommand{\Bd}{\mathcal{L}^{\infty}(\Omega,\mathcal{F})}

\newcommand{\Borel}{\mathbb{B}}

\newcommand{\Om}{\Omega}

\newcommand{\om}{\omega}

\newcommand{\mbu}{\mathbb{u}}

\newcommand{\brv}{\mathcal{L}^{\infty}(\Omega,\mathcal{F})}
\newcommand{\brvG}{\mathcal{L}^{\infty}(\Omega,\mathcal{G})}

\newcommand{\Ep}[1]{\E_\PW\left[#1\right]}
\newcommand{\Eq}[1]{\E_\QW\left[#1\right]}
\newcommand{\Acal}{{\mathcal A}}

\newcommand{\Ecal}{{\mathcal E}}
\newcommand{\Fcal}{{\mathcal F}}
\newcommand{\Gcal}{{\mathcal G}}

\newcommand{\Lcal}{{\mathcal L}}

\newcommand{\Ncal}{{\mathcal N}}

\newcommand{\cm}[1]{\mathfrak{m}\left(#1\right)}
\newcommand{\snorm}[1]{\left\|#1\right\|_{\infty}}

\newcommand{\Linf}{L^{\infty}(\Omega,\mathcal{F},\mathbb{P})}
\newcommand{\LinfG}{L^{\infty}(\Omega,\mathcal{G},\mathbb{P})}

\newcommand{\sigmalg}{$\sigma$-algebra }
\newcommand{\ug}{\mathbb{u}_\mathcal{G} }
\newcommand{\extR}{\left[-\infty, +\infty\right]}

\newcommand{\ima}[1]{\text{Im}_\omega(#1)}
\newcommand{\st}{\middle |}
\newcommand{\dP}{\mathrm{d}\mathbb{P}(\om)}

\newcommand{\nullset}{\Ncal_\succeq}

\newtheorem{theorem}{Theorem}
\newtheorem{assumption}[theorem]{Assumption}

\newtheorem{definition}{Definition}
\newtheorem{lemma}[theorem]{Lemma}
\newtheorem{corollary}[theorem]{Corollary}
\newtheorem{remark}[theorem]{Remark}
\newtheorem{proposition}[theorem]{Proposition}
\newtheorem{example}{Example}

\theoremstyle{remark}

\begin{document}

	\title{On conditioning and consistency for nonlinear functionals}

	\author{Edoardo Berton \thanks{Università Cattolica del Sacro Cuore, Milano, edoardo.berton@unicatt.it} \and Alessandro Doldi \thanks{Università degli Studi di  Firenze, alessandro.doldi@unifi.it} \and Marco Maggis \thanks{ Università degli Studi di  Milano, marco.maggis@unimi.it}}
	
	\date{}

	\maketitle
	
	\begin{abstract}
		\noindent We consider a family of conditional nonlinear expectations  defined on the space of bounded random variables and indexed by the class of all the sub-sigma-algebras of a given underlying sigma-algebra. We show that if this family satisfies a natural consistency property, then it collapses to a conditional certainty equivalent defined in terms of a state-dependent utility function.  This result is obtained by embedding our problem in a decision theoretical framework and providing a new characterization of the Sure-Thing Principle. In particular we prove that this principle characterizes those preference relations which admit consistent backward conditional projections. We build our analysis on state-dependent preferences for a general state space as in Wakker and Zank (1999) and show that their numerical representation admits a continuous version of the state-dependent utility. In this way, we also answer positively to a conjecture posed in the aforementioned paper. 
	\end{abstract}
	
	

	\section{Introduction}
	In the last decades an intensive interplay between Decision Theory and Financial Mathematics led to the development of stochastic portfolio optimization in financial markets, started in the seminal contribution of Merton \cite{Me73} and later followed by a flourishing stream of literature. The maximization of an economic agent's utility can be formulated as a stochastic control problem adopting a preference ordering $\succeq$ on the space of random variables. Such a preference is usually represented as an expected utility and therefore automatically satisfies the so-called Sure-Thing Principle.
	\\ The Sure-Thing Principle (commonly referred to as Axiom P2) began to enjoy popularity after the celebrated book by L.J. Savage \cite{Savage54}, where it was illustrated by the following (to some extent tautological) example
	\begin{center}
		\emph{A businessman contemplates buying a certain piece of property. He considers the outcome of the next presidential election relevant. So, to clarify the matter to himself, he asks whether he would buy if he knew that the Democratic candidate were going to win, and decides that he would. Similarly, he considers whether he would buy if he knew that the Republican candidate were going to win, and again finds that he would. Seeing that he would buy in either event, he decides that he should buy, even though he does not know which event obtains, or will obtain, as we would ordinarily say.} \cite[page 21]{Savage54}.
	\end{center}
	We can rephrase the example as follows: if a decision maker prefers a position $f$ to $g$, either knowing that the event $B$ or $B^C$ (the complement) has occurred, then she should prefer $f$ to $g$ even if she knows nothing about $B$.  As depicted in \cite[page 98]{Gilboa09} ``P2 is akin to requiring that you have conditional preferences, namely, preferences between $f$ and $g$ conditional on $A$ occurring, and that these conditional preferences determine your choice between $f$ and $g$ if they are equal in case $A$ does not occur''. In this paper we explore how this principle is a strongly characterizing property for the existence of conditional (time consistent) backward projections of the initial ordering relation. 
	The validity of the Sure-Thing Principle has been questioned and highly debated as a consequence of the Ellsberg Paradox, (see e.g. \cite[Section 12.3]{Gilboa09}). From this perspective our findings might be seen as arguing  the somewhat excessive strength of this axiom in dynamic decision making.     
	\\We work with a preference $\succeq$ on the space of bounded random variables $\brv$ which satisfies monotonicity and pointwise continuity. We show that for this class of preferences the Sure-Thing Principle is equivalent to the existence of a $\Gcal$-measurable $\succeq$-equivalent (which we shall refer to as conditional Chisini mean \cite{DM23}) for any $f\in\brv$ and $\Gcal$ being a $\sigma$-algebra contained in $\Fcal$ (see Theorem \ref{characterization} below).  The conditional Chisini mean is a nonlinear extension of the standard concept of conditional expectation, which takes inspiration from the original Chisini's intuition \cite{Ch29} that the notion of ``expected value'' should be introduced as the solution of a functional equation. 
	\\Additionally we provide in Theorem \ref{CCE} a representation of the conditional Chisini mean in the form $\ug^{-1}(\Ep{\mbu(f)\middle|\Gcal})$, where $\Ep{\mbu(f)\middle|\Gcal}$ is the conditional expectation of the random outcome $\omega\mapsto \mbu(\omega,f(\omega))$ and $\ug(\omega,x)=\Ep{\mbu(x)\middle|\Gcal}(\omega)$ is the $\Gcal$-measurable projection of the state-dependent utility $\mbu$. This representation can be easily compared to the general formulation proved in \cite{CVMM}, Lemma 5.2, under additional law invariance assumptions, which traces back to the celebrated Nagumo, de Finetti, Kolmogorov Theorem (see \cite{HLP52} for a review). It is also worth mentioning that Theorem \ref{CCE} leads to a theoretical foundation of \cite{FM11cce}, where nonlinear maps are \emph{a priori} assumed to be in the form of conditional certainty equivalents.
	\\This result is obtained by solving an open question posed in \cite{WZ99} regarding the continuity (and joint measurability) of the state-dependent utility appearing in the representing functional of the preference relation (see the statement of Theorem \ref{formaintegrale} in Section \ref{open:question}). From a mathematical viewpoint Theorem \ref{thm:general_thm} is the most triggering result, as it requires a quite structured and involved proof. Nevertheless such a result has the immediate benefit of clarifying the applicability of state-dependent preferences in optimization problems such as the maximization over a set of strategies.
	
	\subsection{Representation of time consistent nonlinear expectations}
	
	The representation of the conditional Chisini mean depicted so far allows to embed our theoretical results in the framework of nonlinear expectations. This deepens the understanding of the role of the Sure-Thing Principle in time consistent conditioning. 
	\\ Kupper and Schachermayer \cite{KuSch09} proved the noteworthy property that, under suitable structural assumptions, every time consistent family of strictly monotone, continuous and law invariant nonlinear maps take the form of a conditional certainty equivalent
	\begin{equation}\label{KuScha} u^{-1}\left(\Ep{u(X)\mid \Fcal_t}\right),\end{equation}
	where $u$ is a strictly increasing function on a real interval. This result can be seen as a dynamic version of the aforementioned Nagumo-de Finetti-Kolmogorov Theorem.  
	\\ A similar perspective can be observed in the recent stream of literature regarding the phenomenon of ``collapse to the mean'' for certain classes of functionals as in \cite{BKMS21} for the law invariant case. Indeed the findings of our paper relate significantly to Theorem 6.1 in \cite{Delbaen21},  where the assumption of law invariance is replaced with commonotonicity. 
	In the realm of nonlinear expectations which are possibly not law invariant, the most important contribution comes from the theory of $g$-expectations related to Backward Stochastic Differential Equations, initiated and developed in \cite{Peng97, CoquetPeng02, Peng04}. The main drawback of such theory is that the notion of ``conditionability'' is constrained to the use of the Brownian filtration. Therefore in general it is not possible to find a conditional nonlinear expectation for finitely generated $\sigma$-algebras, that is also consistent with the considered $g$-expectation.  
	\\ We overcome this issue by taking into consideration a non atomic\footnote{Actually, requiring the existence of three disjoint sets in $\Fcal$ with positive probability is sufficient to prove Theorem \ref{repr:nonlinear}} probability space $(\Omega,\Fcal,\PW)$, the collection of sub-$\sigma$-algebras of $\Fcal$
	$$
	\Sigma := \left\{ \Gcal \middle| \Gcal \text{ is a $\sigma$-algebra},\; \Gcal \subseteq \Fcal \right\}, 
	$$
	and a family of conditional nonlinear maps $\{\Ecal_{\Gcal}\}_{\Gcal\in\Sigma}$, where $\Ecal_{\Gcal}:\Linf\to \LinfG$,
	with $\Linf$ (resp. $\LinfG$) the quotient space of bounded random variables with respect to $\PW$-almost sure equality. We denote with $\Ecal_{0}$ the map obtained choosing $\Gcal=\sigma(\emptyset)$, the trivial $\sigma$-algebra.
	Moreover we indicate by capital letters the elements in $X\in\Linf$, and by $\ind_A$ the indicator functions\footnote{Here we intend the equivalence class of such indicators. Later on we shall work also without quotienting adopting the same symbol to avoid heavy notations.}.
	\begin{definition}\label{def:nonlinear}
		We say that $\Ecal_{\Gcal}:\Linf\to\LinfG$ is a conditional nonlinear expectation if for any $X\in\Linf$, $A\in\Gcal$ we have $\Ecal_{\Gcal}(X\ind_A)=\Ecal_{\Gcal}(X)\ind_A$. 
		\\A family of conditional nonlinear expectations $\{\Ecal_{\Gcal}\}_{\Gcal\in\Sigma}$ is time consistent if
		\begin{equation}
			\label{tower_prop} \Ecal_0(\Ecal_{\Gcal}(X))=\Ecal_0(X) \text{ for any } \Gcal\in\Sigma \text{ and } X\in\Linf.
		\end{equation}
		
	\end{definition}
	
	\begin{assumption}\label{ass:E0}
		We assume that $\Ecal_0$ is strictly monotone (on dicotomic random variables), i.e. for any $x,y\in\R$, $Z \in \Linf$ and $A\in \Fcal$ with $\PW(A)>0$
		\begin{equation} 
			\label{strict:mon} x<y \text{ implies } \Ecal_{0}(x\ind_A+Z\ind_{A^c})<\Ecal_{0}(y\ind_A+Z\ind_{A^c}),
		\end{equation}
		and pointwise continuous, i.e. for any bounded sequence $X_n$ converging $\PW$-almost surely to $X$
		\begin{equation}\label{point:cont} \lim_{n\to\infty}\Ecal_0(X_n)=\Ecal_0(X).
		\end{equation}
	\end{assumption}

	
	The following representation result can be seen as an enhanced version of  \eqref{KuScha} dropping the law invariance. Similarly to \cite{KuSch09} and differently from Lemma 5.2 in \cite{CVMM}, one of the most interesting aspects of our result is that the Sure-Thing Principle needs not to be explicitly assumed, but is rather deduced from the time consistency of the maps involved.
	
	\begin{theorem}\label{repr:nonlinear}
		$\{\Ecal_{\Gcal}\}_{\Gcal\in\Sigma}$ is a family of time consistent conditional nonlinear expectations such that $\Ecal_0$ satisfies Assumption \ref{ass:E0} if and only if \footnote{See the beginning of Section \ref{proof:thm2} for precise definitions}  
		\[\Ecal_{\Gcal}(X)=\ug^{-1}(\Ep{\mbu(X)\st\Gcal})\quad \text{ for every }X\in\Linf,\]
		for some  $\mbu:\Omega\times\R\to \R$ which is jointly measurable, with $\mbu(\omega,0)=0$ and $\mbu(\omega,\cdot)$  continuous and strictly increasing for every $\omega\in\Omega$, $\mbu(\cdot,x)$ integrable for any $x\in \R$ and $\ug(\cdot,x)=\Ep{\mbu(x)\st\Gcal}$. 
	\end{theorem}
	
	
	\subsection{Notations and preliminaries} \label{preliminaries}
	
	Throughout this paper $(\Om, \Fcal)$ denotes a measurable space, with $\Om$ being a general state space and $\Fcal$ a $\sigma$-algebra of subsets of $\Omega$, which we refer to as events. We shall denote by $\R$ the set of real numbers, which will always be endowed with the Borel $\sigma$-algebra $\Borel_{\R}$. By $\overline{\R}$ we indicate the extended real line $\extR$ and with $\Borel_{\overline{\R}}$ the associated Borel $\sigma$-algebra. We use $<, >$ for strict inequalities between real numbers. Similarly, $\subset, \supset$ are used to indicate proper subsets. In case equalities are not excluded, we use $\leq, \geq$ and $\subseteq, \supseteq$, respectively. We consider, as conventional, $\inf \varnothing = +\infty$. We denote by $\mathcal{L}(\Omega,\Fcal)$ the space of
	$\Fcal$-measurable functions taking values in $\R$. We call random variables the elements $f \in \mathcal{L}(\Omega,\Fcal)$ and we denote with $\Bd$ the collection random variables $f$ such that for some $k \geq 0$ $|f(\omega)|\leq k$ for every $\omega\in\Omega$. We consider on $\mathcal{L}(\Omega,\Fcal)$ and $\Bd$ the natural pointwise order $f \geq g$ if and only if $f(\omega) \geq g(\omega)$ for all $\omega \in \Om$. The space $\Bd$ is a Banach space with respect to the usual sup norm $\|f\|_{\infty}=\sup_{\om \in \Om}|f(\om)|$.
	Given any $A\in \Fcal$, we define $\ind_A$ as the element of $\Bd$ that takes value 1 if $\omega \in A$ and 0 otherwise and denote by $\sigma(A)$ the $\sigma$-algebra generated by $A$. For any pair of random variables $f,g \in \Bd$ we denote with $f\ind_A + g\ind_{A^c}$ the random variable that agrees with $f$ whenever the event $A$ occurs and with $g$ otherwise.
	Working with the product space $\Omega\times \R$, we shall make use of the product  $\sigma$-algebra $\Fcal\otimes \Borel_{\R}$ i.e. the \sigmalg generated by elements of the form $F\times B$ with $F\in\Fcal$ and $B\in \Borel_{\R}$. Whenever we fix a probability measure $\PW$ on $\Fcal$, we call $\Space$ a probability space.
	A (finite) signed measure $\nu$ on $(\Omega,\Fcal)$ is said to be dominated by $\PW$ ($\nu \ll \PW$) if, for any $A \in \Fcal$, $\PW(A) = 0$ implies $\abs{\nu}(A) = 0$, where $\abs{\nu}$ is the total variation measure of $\nu$. By the Hahn decomposition theorem, $\nu\ll\PW$ is equivalent to: $\PW(A) = 0$ implies ${\nu}(A) = 0$. For a given probability $\PW$ on $(\Omega,\Fcal)$ we shall denote the space of integrable functions by $\Lcal^1(\Om, \Fcal,\PW)=\left\{f\in \Lcal(\Om,\Fcal) : \Ep{|f|}<\infty \right\}$. \\For any $f\in \Lcal^1(\Om, \Fcal,\PW)$ and a sub $\sigma$-algebra $\Gcal$, $\Ep{f\middle | \Gcal}$ is the conditional expectation i.e.
	\begin{equation}\label{cond:exp}
		\Ep{f \middle| \Gcal}=\{g\in \Lcal^1(\Om, \Gcal,\PW) : \Ep{g\ind_A}=\Ep{f\ind_A}\;\forall\,A\in\Gcal \}.\end{equation}
	In the following we shall assume familiarity with the properties of conditional expectation.
	
	\medskip
	
	A preference relation is a binary relation $\succeq$ on $\Bd$ which satisfies completeness and transitivity and induces a (strict) preference order $\succ$ by $f\succ g$ if and only if $f \succeq g$ but $g\nsucceq f$. In addition we set $g\sim f$ whenever $g\succeq f$ and $f\succeq g$. 
	We denote by
	\begin{equation}
		\label{defnulls}
		\Ncal_\succeq := \left\{ A \in \Fcal : f\ind_A + g\ind_{A^c} \sim g \;\; \forall f,g \in \brv \right\}
	\end{equation}
	the set of null events (induced by the preference $\succeq$). 
	\\We present a few additional properties that the preference relation might enjoy.
	
	\begin{description}
		\item[(SM)]\label{item:strictmon} $\succeq$ is strictly monotone if $x\ind_A +
		f\ind_{A^c} \succ y\ind_A + f\ind_{A^c}$, for
		all non null events $A\in \Fcal$, for all $f \in \brv$ and
		constant outcomes $x,y\in\R$ with $x > y$;
		\item [(ST)] $\succeq$ satisfies the Sure-Thing Principle whenever, considering arbitrary $f, g, h \in
		\brv$ and $A\in \Fcal$ such that $f\ind_A +
		h\ind_{A^c} \succeq g\ind_A + h\ind_{A^c}$, we have $f\ind_A + \tilde{h}\ind_{A^c}
		\succeq g\ind_A + \tilde{h} \ind_{A^c}$ for every $\tilde{h} \in \brv$;
		\item[(PC)] $\succeq$ is pointwise continuous if, for any uniformly bounded
		sequence $(f_n)_n\subseteq \brv$, such that
		$f_n(\omega)\rightarrow f(\omega)$ for any $\omega\in\Omega$ then
		for all $g \in \brv$ such that $g \succ f$ (resp. $f \succ g$)
		there exists $N \in \mathbb{N}$ such that $g \succ f_n$ (resp. $f_n \succ
		g$) for every $n > N$.
	\end{description}
	
	We say that $\succeq$ admits a numerical representation if there exists a functional $T: \Bd \to \R$ such that for any given $f, g \in \Bd$ 
	\begin{equation}\label{repr:T}
		T(f) \geq T(g) \quad \text{ if and only if } \quad f \succeq g.
	\end{equation}
	

	\section{Conditionable functionals and preferences}
	In \cite{DM23} the notion of conditional Chisini mean was introduced as the solution of an infinite dimensional system of equations as we now recall in the following definition. Such a notion is intrinsically related to the structural properties of the functional $T$ involved, which in this paper will be the representing functional of $\succeq$. 
	
	\begin{definition}\label{conditionable} For a given functional $T:\brv\to\R$, we say that  $T$ is conditionable on a sub $\sigma$-algebra $\Gcal\subseteq \Fcal$  if for any $f\in\brv$ there exists $g\in\brvG$ such that 
		\begin{equation}\label{eq:conditioning_family}
			T(f\ind_A)= T(g\ind_A) \quad \forall\,A\in \Gcal. 
		\end{equation}
		We say that $T$ is conditionable if it is conditionable on any $\Gcal\in\Sigma$ and we call conditional Chisini mean $\cm{f\middle|\Gcal}$ the family of solutions $g$ of Eq. \eqref{eq:conditioning_family}.  
	\end{definition}
	
	If compared with \eqref{cond:exp}, it evidently appears that $\cm{f\middle|\Gcal}$ is the nonlinear counterpart of the standard conditional expectation. Although linearity is lost, nonlinear expectation still enjoys few useful properties of its linear counterpart such as, for instance, the ``tower'' and ``taking out what is known'' properties. Before the abstract formulation in \cite{DM23}, the problem of existence of nonlinear conditional expectation was already addressed in \cite{CoquetPeng02}, yet the provided solution relies on a Brownian framework which in turn demands for more technical assumptions. 
	\\ Embedding such a  problem in a decision theoretical setup, we could rephrase the issue as follows: given a preference order $\succeq$ on $\brv$ then the certainty equivalent of $f\in\brv$ is $a\in\R$ such that $a\sim f$. Similarly if the agent is provided with additional information, represented by the $\sigma$-algebra $\Gcal \subset \Fcal$, then the $\Gcal$-\emph{conditional certainty equivalent}  of $f\in\brv$ is any $g\in \brvG$ such that    
	\begin{equation}\label{eq:chisini_problem}
		f\ind_A\sim g\ind_A  \quad \forall \;A\in\Gcal.
	\end{equation}
	Loosely speaking,  $g\in\brvG$ can be considered as a $\Gcal$-measurable counterpart of the $\Fcal$-measurable function $f\in \brv$ (see also \cite{FM11cce} and \cite{MM21} for applications of the conditional certainty equivalent).
	\\ Since this paper deals with state-dependent preferences we shall directly refer to the conditional Chisini mean as the class
	\begin{equation}\label{eq:sol_chisinimean}
		\cm{f\middle|\Gcal}=\left\{g\in\brvG\mid f\ind_A\sim g\ind_A\;\forall\,A\in\Gcal\right\},
	\end{equation}
	keeping in mind that Eq. \eqref{eq:sol_chisinimean} translates automatically to the framework in \cite{DM23} as soon as we characterize a preference $\succeq$ through a representing functional (as defined in \eqref{repr:T}).  
	
	\begin{remark}(SM) and (PC) of the preference $\succeq$ are sufficient to show the existence of a preference representing functional $T$ (see the proof of Proposition \ref{prop:exist_unique} for the further details).
	\end{remark}

	We are now ready to state the first main result of this work, which links the Sure-Thing Principle to the conditionability of preferences and their representing functionals. The proof is postponed to Section \ref{sec:proof_characterization}.
	
	\begin{theorem}\label{characterization}
		Let $\succeq$ be a preference order satisfying (SM) and (PC). Then the following are equivalent
		\begin{enumerate}[(i)]
			\item \label{prop:ST} $\succeq$ satisfies (ST);
			\item \label{prop:cond} any representing functional $T$ for $\succeq$ is conditionable;
			\item \label{prop:condA} any representing functional $T$ for $\succeq$ is conditionable on\footnote{Where $\sigma(A) = \left\{\Omega, \varnothing, A , A^C\right\}$} $\sigma(A)$ for any $A\in \Fcal$.
		\end{enumerate}
	\end{theorem}
	
	We propose the following definition of regularity that is well suited for our purposes. In particular, it encompasses both the notion of continuity and that of joint-measurability that we seek to obtain for the state-dependent utility.
	
	\begin{definition}\label{repr:CCE} For any $\Gcal\in\Sigma$ we say that $\varphi:\Om \times \R \to \R$ is a $\Gcal$-regular function if   
		\begin{itemize}
			\item $\varphi(\omega, \cdot)$ is continuous and strictly increasing on $\R$ for every $\omega \in \Om$;
			\item $\varphi(\om, 0) = 0$ for every $\om \in \Om$;
			\item $\varphi$ is $\Gcal\otimes \Borel_{\R}$ measurable.
		\end{itemize}
	\end{definition}
	
	\begin{definition}
		We say that representability holds for conditional Chisini means if there exist an $\Fcal$-regular function $\mbu$, a probability measure $\PW$ on $\Fcal$ and, for every $\sigma$-algebra $\Gcal \subseteq \Fcal$, a $\Gcal$-regular function $\ug$ such that $\mbu(\cdot, x)$ and $\ug(\cdot, x)$ are integrable for every $x \in \R$ and for any $f \in \brv$ and $g \in \cm{f\middle|\Gcal}$ we can write\footnote{More details can be found in the statement of Theorem \ref{CCE}.}
		\[g(\omega) = \ug^{-1}(\omega, h(\omega)) \quad \forall \,\omega\in\Omega,\] 
		up to a choice of a version $h$ of $\Ep{\mbu(\cdot,f)\middle| \Gcal}$ . 
	\end{definition}
	

	\begin{corollary}\label{cor:representability}
		Let $\succeq$ be a preference order satisfying (SM) and (PC). Whenever $\Fcal$ contains at least three disjoint non-null events, condition (\ref{prop:ST}) of Theorem \ref{characterization} is equivalent to
		\begin{enumerate}[(i)]
			\setcounter{enumi}{3}
			\item \label{prop:CCE} conditional Chisini means are representable.
		\end{enumerate}   
	\end{corollary}

	\section{Solution to an open question in Wakker and Zank (1999)}\label{open:question}
	
	The proofs of the results stated in the previous sections mostly rely on an enhanced version of the preference representation result due to Wakker and Zank \cite{WZ99}, which we are now going to present in detail.
	For a general state space $(\Omega,\Fcal)$, Wakker and Zank \cite{WZ99} provide a characterization of those preference relations $\succeq$ on the space of bounded random acts which admit a numerical representation through an additively decomposable functional (\cite{WZ99}, Theorem 11). This result extends to an infinite dimensional setup the theory developed by Debreu \cite{De60}. By adopting the stronger notion of pointwise continuity of $\succeq$, Theorem 11 in \cite{WZ99} can be specialized to the following integral representation form.
	
	\begin{theorem}[\cite{WZ99}, Theorem 12 and \cite{CL06}, Theorem 5]\label{formaintegrale}
		Assume that $\Fcal$ contains at least three disjoint non-null events and $\succeq$ satisfies (SM), (PC) and (ST). Then there exist:
		\begin{enumerate}[(i)]
			\item a probability measure $\PW$ on $(\Om, \Fcal)$ such that $\nullset = \left\{A \in \Fcal \st \PW(A) = 0 \right\}$,
			\item\label{wz_u} a function $\mbu:\Om \times \R \to \R$ with $\mbu(\om, \cdot)$ strictly increasing on $\R$ for every $\omega \in \Om$ and $\mbu(\cdot, x) \in \Lcal^1(\Om, \Fcal, \PW)$ for every $x \in \R$,
		\end{enumerate}
		such that given $f,g \in \Bd$ we have 
		\begin{equation}\label{repr:functional}
			f \succeq g \text{ if and only if } \int_\Om \mbu(\omega, f(\omega)) \dP \geq \int_\Om \mbu(\omega, g(\omega)) \dP
		\end{equation}
	\end{theorem}
	
	
	The previous theorem leaves an interesting issue unsolved as declared by the authors
	\cite{WZ99}, page 29:  
	\begin{center}``[...] Pointwise continuity is not overly restrictive
		because continuity of each state-dependent utility implies pointwise continuity. Whether the reversed implication holds, i.e., whether continuity of the functions $\mbu$, in Statement (i) of Theorem\footnote{Referred to the original statement in \cite{WZ99}} 12  holds (after appropriate modification on null events), is an open question to us.''    
	\end{center}
	
	To the best of our knowledge this question did not find an answer in the subsequent literature.
	Among the principal aims of the present paper, Theorem \ref{thm:general_thm} stated below provides the construction of a pointwise continuous version of the state-dependent utility $\mbu$\footnote{Notice that in Appendix A of \cite{MM21} a somewhat cumbersome notion of continuity was introduced to overcome some technical hurdles in the proof of the main result. Guaranteeing the existence of a continuous $\mbu$ allows to easily circumvent issues of this kind.}.
	
	
	\begin{theorem}\label{thm:general_thm} Assume that $\Fcal$ contains at least three disjoint non-null events and $\succeq$ satisfies (SM), (PC) and (ST). Then $\mbu$ in Theorem \ref{formaintegrale} can be chosen to be $\Fcal$-regular.
	\end{theorem}
	

	\section{Explicit representation of conditional Chisini means}\label{statement}
	
	In this section we apply Theorem \ref{thm:general_thm} to obtain an explicit formula for conditional Chisini means, which is the cornerstone of the proofs of Theorem \ref{repr:nonlinear} and of Corollary \ref{cor:representability}. 

	\begin{theorem}
		\label{CCE}
		Assume that the hypotheses of Theorem \ref{thm:general_thm} are satisfied and consider a resulting pair $(\mbu,\PW)$ with $\mbu$ being $\Fcal$-regular. Then, for any given a sub-sigma algebra $\Gcal\subseteq \Fcal$, $\cm{f\middle|\Gcal} \neq \varnothing$ and there exist
		\begin{itemize}
			\item a $\Gcal$-regular function $\ug:\Om \times \R \to \R$, with 
			\begin{enumerate}[(i)]
				\item $\ug(\cdot, x) \in \Lcal^1(\Om, \Gcal, \PW)$ for every $x \in \R$, 
				\item $\ug(\cdot,x)$ is a version of $\Ep{\mbu(\cdot,x)\middle|\Gcal}$ for every $x\in\R$.
			\end{enumerate}
			\item a version $g$ of $ \Ep{\mbu(\cdot, f(\cdot))\middle|\Gcal}$ such that
			\begin{equation}
				\label{eq:chisini}
				\Phi_\Gcal(\cdot, g(\cdot))\in\cm{f\middle|\Gcal},        
			\end{equation}
		\end{itemize}
		where $\Phi_\Gcal:\Omega\times \R\rightarrow [-\infty,+\infty]$ defined by
		\[\Phi_\Gcal(\omega,x):=\inf\{y\in\R\mid \ug(\omega,y)>x\}\]
		is $\Gcal\otimes\Borel_\R$ measurable.
	\end{theorem}
	
	We write informally \eqref{eq:chisini} as 
	$$\cm{f\st\Gcal}=\ug^{-1}(\Ep{\mbu(f)\st\Gcal}),$$
	and recall that, under the assumptions of Theorem \ref{CCE}, $\cm{f\st\Gcal}$ is unique up to $\succeq$-null events (see \cite{DM23}, Theorem 3.2 and Proposition \ref{prop:exist_unique} in the following section).
	
	\medskip
	
	To clarify the statement of Theorem \ref{CCE}, it seems important at this stage to point out that the use of $\Phi_\Gcal$ in \eqref{eq:chisini} conceals some technical subtleties. To begin with, the dependence of the image of $\ug$ on $\om$ requires cautious handling of its domain. Additionally, for $f \in \brv$, the nature of $g\in\Ep{\mbu(f)\st\Gcal}$ implies that it may not be an element of $\brvG$ itself.  While the case of a finitely generated $\sigma$-algebra is rather simple (as we shall depict in Example \ref{example:CC} below), we shall dedicate a comprehensive treatment to the issue for an arbitrary $\sigma$-algebra $\Gcal$ in the proof postponed to Section \ref{proof:CCE}.
	
	\begin{example}\label{example:CC}
		Consider $\{A_1, ..., A_n\}$ being a finite partition of $\Omega$ such that $\PW(A_i) > 0$ for every $i \leq n$ and the $\sigma$-algebra $\Gcal$ generated by this partition. Then every $\Gcal$-measurable function $g$ can be written as a finite linear combination of indicator functions of the atoms $A_i$. In particular, for $\om \in A_i$, $\ug(\om, x)$ coincides with
		\begin{equation*}
			u_i(x)= \frac{1}{\PW(A_i)} \int_{A_i} \mbu(\om, x)\dP,
		\end{equation*} $u_i$ being finite since $\mbu(\cdot, x) \in \Lcal^1(\Om, \Fcal, \PW)$. Clearly $\ug(\cdot, x) \in \Lcal^1(\Om, \Gcal, \PW)$ for all $x \in \R$, $\ug(\om, x) = \Ep{\mbu(\cdot, x)|\Gcal}(\om)$ and $x\mapsto \ug(\om, x)$ is increasing and continuous for all $\om \in \Om$. 
		\\ For any fixed $A_i$ we notice that the image of $\ug(\om, \cdot)$, namely $R_i := \mathrm{Im}(\ug(\om, \cdot))$ does not depend on $\omega\in A_i$. Given the generalized inverse $\Phi_\Gcal: \Om \times \R \to \R$ defined in Theorem \ref{CCE}, it is crucial to observe that the map $x \mapsto \Phi_\Gcal(\om, x)=\ug^{-1}(\om, x)$ is well defined as a standard inverse function for $\om \in A_i$ and $x \in R_i$ (as a result of strict monotonicity of $\ug(\om, \cdot)$).
		
		In particular if we now choose any $f \in \brv$ and set $h:=\Ep{\mbu(\cdot, f(\cdot))|\Gcal}$ then $h(\omega)\in R_i$ for $\om \in A_i$ and indeed
		\begin{equation*}
			(\ug \circ \Phi_\Gcal)(\om, h(\om)) = h(\om).
		\end{equation*}	
		For $f \in \brv$ denote $\mathfrak{m}_i = u_i^{-1} \left(\frac{1}{\PW(A_i)} \int_{A_i} \mbu(\om, f(\om))\dP \right)$ so that for any $\omega\in A_i$ we have 
		$$\ug(\omega,\mathfrak{m}_i)= \Ep{\mbu(\cdot, f(\cdot))|\Gcal}(\omega).$$ 
		We conclude the example noting that the Chisini mean $\cm{f\st\Gcal}$ for the functional $T(f)=\int_\Om \mbu(\omega, f(\omega)) \dP$ consists, in the current setup, of a single random variable $\cm{f\st\Gcal}$ which assumes the value 
		\begin{equation*}
			u_i^{-1} \left(\frac{1}{\PW(A_i)} \int_{A_i} \mbu(\om, f(\om))\dP \right)
		\end{equation*}
		on the set $A_i$. This can be readily verified following the procedure in Section \ref{proof:CCE}, step 4.
	\end{example}
	
	\section{Proofs of main results}
	
	We stress that the order of the results as presented in the main body of the paper is chosen to put the right focus on the meaning of the results.This section instead is structured in such a way that the logical flow of the proofs is respected. 
	
	\subsection{Proof of Theorem \ref{characterization}}\label{sec:proof_characterization}
	
	This section is dedicated to the proof of Theorem \ref{characterization}. We first show that the conditional Chisini mean exists for a preference satisfying (SM), (PC) and (ST) and that it is unique up to $\succeq$-null events. Then we provide a ``taking out what is known'' property for the conditional Chisini mean before focusing on the proof of the main result.
	
	\begin{proposition}\label{prop:exist_unique}
		Let $\succeq$ be a preference order satisfying (SM), (PC), (ST) and consider a sub-$\sigma$-algebra $\Gcal \subseteq \Fcal$ and $f \in \brv$. Then $\cm{f\middle|\Gcal} \neq \varnothing$ and for any $g, \tilde{g} \in \cm{f\middle|\Gcal}$ it holds that $\left\{g \neq \tilde{g} \right\} \in \Ncal_\succeq$.
	\end{proposition}
	\begin{proof}
		First, observe that for any preference order satisfying (SM) and (PC) there exists a representing functional $T:\brv \to \R$. Indeed, (SM) and (PC) impliy the pointwise monotonicity of $\succeq$ on $\brv$ (see Lemma \ref{lem:ptwise_mon}). Consequently, for any $f \in \brv$ the sets $\underline{F} := \{ x \in \R: f \succeq x\}$ and $\overline{F} := \{ x \in \R: x \succeq f\}$ are non empty closed halflines in $\R$ such that $\overline{F} \cup \underline{F} = \R$. Then $\overline{F} \cap \underline{F} = \left\{z\right\}$, with $z$ being the (unique) certainty equivalent of $f$, i.e. $f \sim z$, and we can define $T(f) := z$. \\
		We show that the hypotheses of \cite{DM23}, Theorem 3.2 hold for such a functional $T$. Indeed, (SM), (PC) and (ST) of $\succeq$ imply that $T$ satisfies\footnote{We adopt the notation of \cite{DM23}} ($\Gcal$-Mo), ($\Gcal$-PC), ($\Gcal$-QL) and ($\Gcal$-NB) for any $\sigma$-algebra $\Gcal \subseteq \Fcal$. It remains to show that $T$ satisfies ($\Gcal$-PS). Consider $f \in \brv$, $A_1, A_2 \in \Gcal$ with $A_1 \cap A_2 = \varnothing$ and let $x_1, x_2 \in \R$ be such that $T(f\ind_{A_i}) = T(x_i\ind_{A_i})$ for $i=1,2$. Define $g_1 := x_1\ind_{A_1}$ and $g_2 := f \ind_{A_2}$. By (ST) we obtain $T(f \ind_{A_1} + g_2 \ind_{A_1^c}) = T(x_1 \ind_{A_1} + g_2 \ind_{A_1^c})$ and similarly we find  $T(f \ind_{A_2} + g_1 \ind_{A_2^c}) = T(x_2 \ind_{A_2} + g_1 \ind_{A_2^c})$. Since $A_2^c \cap A_1 = A_1$ and $A_1^c \cap A_2 = A_2$, it follows that $T(f \ind_{A_1} + f \ind_{A_2}) = T(x_1 \ind_{A_1} + f \ind_{A_2})$ and $T(f \ind_{A_2} + x_1 \ind_{A_1}) = T(x_2 \ind_{A_2} + x_1 \ind_{A_1})$, whence $T(f \ind_{A_1} + f \ind_{A_2}) = T(x_1 \ind_{A_1} + x_2 \ind_{A_2})$.
		We can therefore apply \cite{DM23}, Theorem 3.2 which in turn provides $\cm{f\middle|\Gcal} \neq \varnothing$ and the desired uniqueness.
	\end{proof}
	
	\begin{lemma}\label{lem:taking_out}
		Let $\succeq$ be a preference order satisfying (SM), (PC), (ST) and consider a sub-$\sigma$-algebra $\Gcal \subseteq \Fcal$ and $f \in \brv$. For $A \in \Gcal$, we have
		\begin{equation}\label{eq:taking_out}
			\cm{f\ind_A\middle|\Gcal} = \cm{f\middle|\Gcal}\ind_A
		\end{equation}
		meaning that $g\in \cm{f\ind_A\middle|\Gcal}$ if and only if $g = \overline{g}\ind_A$ for some  $\overline{g}\in \cm{f\middle|\Gcal}$. 
	\end{lemma}
	
	\begin{proof}
		Consider $g \in \cm{f\ind_A\middle|\Gcal}$. By definition we have that $T(g\ind_B) = T((f\ind_A) \ind_B)$ for all $B \in \Gcal$. As $\ind_A \ind_B = \ind_{A \cap B}$ and $A \cap B \in \Gcal$, we can write $T((f\ind_A) \ind_B) = T(f\ind_{A\cap B}) = T(\overline{g}\ind_{A\cap B}) = T((\overline{g}\ind_A)\ind_B)$, with $\overline{g} \in \cm{f\middle|\Gcal}$, and hence it follows that $g = \overline{g}\ind_A$. \\
		Now let $\overline{g} \in \cm{f\middle|\Gcal}$ and fix $A \in \Gcal$. For all $G \in \Gcal$ it holds that $A \cap G \in \Gcal$, so that in particular $T(f\ind_{A \cap G}) = T(\overline{g}\ind_{A \cap G})$. It then follows for all $G \in \Gcal$ that $T(\overline{g}\ind_{A \cap G}) = T(f\ind_{A \cap G}) = T((f\ind_A) \ind_G) = T(g \ind_G),$
		with $g \in \cm{f\ind_A\middle|\Gcal}$. We conclude that $\overline{g}\ind_A = g$.
	\end{proof}
	
	\noindent We further observe that, as $\cm{0 \middle|\Gcal}= 0$ by definition, for any $A \in \Gcal$ property \eqref{eq:taking_out} is equivalent to  
	\[\cm{f\ind_A + h \ind_{A^c}\middle|\Gcal} = \cm{f\middle|\Gcal}\ind_A + \cm{h\middle|\Gcal}\ind_{A^c}.\]
	
	\begin{proof}[Proof of Theorem \ref{characterization}] We begin by showing that Item (\ref{prop:ST}) implies (\ref{prop:cond}). To this aim, note that the hypotheses of Proposition \ref{prop:exist_unique} are satisfied and hence for any $\Gcal \subseteq \Fcal$ there exists a function $g \in \brvG$ that satisfies Eq. \eqref{eq:chisini_problem}.\\
		Item (\ref{prop:condA}) follows directly from Item (\ref{prop:cond}).\\ 
		We conclude the proof showing that (\ref{prop:condA}) implies (\ref{prop:ST}). Consider $f_1, f_2, f, h \in \brv$ and observe that, by definition of $\nullset$, whenever $A \in \nullset$ then $f\ind_A + h\ind_{A^c} \sim h$, so that $f_1\ind_A + h\ind_{A^c}  \succeq f_2\ind_A + h\ind_{A^c}$ implies $f_1\ind_A + \overline{h}\ind_{A^c} \succeq f_2\ind_A + \overline{h}\ind_A^c$ for any $\overline{h} \in \brv$. Similarly whenever $A^c \in \nullset$. Suppose $A, A^c \in \Fcal \setminus\Ncal_\succeq$ and consider $\Gcal:= \sigma(A)$. First notice that, by construction of $\Gcal$ and uniqueness, for any $f \in \brv$ the conditional Chisini mean $\cm{f\middle|\Gcal}$ collapses to a single $\Gcal$-measurable random variable. In particular, $\cm{f\middle|\Gcal} = x \ind_A + y \ind_{A^c}$ for some $x,y \in \R$. \\
		Let $T:\brv \to \R$ be a representing functional for $\succeq$ conditionable on $\Gcal = \sigma(A)$ and take $f_1, f_2,h \in \brv$ such that $T(f_1 \ind_A + h\ind_{A^c}) \geq T(f_2 \ind_A + h\ind_{A^c})$. Conditionability of $T$ ensures that there exists $g_i \in \brvG$ such that $T(g_i \ind_G) = T((f_i \ind_A + h \ind_{A^c})\ind_G)$ for any $G \in \Gcal$ and $i = 1,2$. As we previously pointed out, for this choice of $\Gcal$ the conditional Chisini mean consists of a single random variable. Hence we denote with $\cm{f_i \ind_A + h \ind_{A^c}\middle|\Gcal}$ the solution $g_i$ to the previous equation. By Lemma \ref{lem:taking_out}, we have that for $i = 1,2$
		$$\cm{f_i \ind_A + h \ind_{A^c}\middle|\Gcal} = \cm{f_i\middle|\Gcal}\ind_A + \cm{h \middle|\Gcal}\ind_{A^c},$$
		and recalling that any $\Gcal$-measurable function is of the form $x\ind_A + y \ind_{A^c}$ for some $x,y \in \R$, we obtain that $\cm{f_i\middle|\Gcal} = x_i$ on $A$ and analogously $\cm{h\middle|\Gcal} = \overline{y}$ on $A^c$ for some $x_i, \overline{y} \in \R$. Now observe that $T(f_1 \ind_A + h \ind_{A^c}) \geq T(f_2 \ind_A + h \ind_{A^c})$ implies $x_1 \geq x_2$. Assume, by way of contradiction, that $x_1 < x_2$, then (SM) yields $T(x_1 \ind_A  + \overline{y}\ind_{A^c}) < T(x_2 \ind_A  + \overline{y}\ind_{A^c})$. However, 
		$$T(f_1 \ind_A + h \ind_{A^c}) = T(x_1 \ind_A  + \overline{y}\ind_{A^c}) < T(x_2 \ind_A  + \overline{y}\ind_{A^c}) = T(f_2 \ind_A + h \ind_{A^c}),$$ 
		hence a contradiction. \\
		Keeping $f_1$ and $f_2$ constant and replacing $h$ with an arbitrary $k \in \brv$ we find $\cm{f_i \ind_A + k \ind_{A^c} \middle|\Gcal} = x_i \ind_A  + z \ind_{A^c}$ for $i=1,2$, with $z = \mathfrak{m}(k|\Gcal)$ on $A^c$. As $x_1 \geq x_2$, (SM) implies
		$$
		T(f_1 \ind_A + k \ind_{A^c}) = T(x_1 \ind_A  + z\ind_{A^c}) \geq T(x_2 \ind_A  + z\ind_{A^c}) =  T(f_2 \ind_A + k \ind_{A^c}).
		$$
		Since $k$ is arbitrary, we get in fact that $T(f_1 \ind_A + h \ind_{A^c}) \geq T(f_2 \ind_A + h \ind_{A^c})$ for one $h\in\brv$ entails $T(f_1 \ind_A + \overline{h} \ind_{A^c}) \geq   T(f_2 \ind_A + \overline{h} \ind_{A^c})$ for all $\overline{h} \in \brv$, which in turn implies that $\succeq$ satisfies (ST).
	\end{proof}
	
	\subsection{Proof of Theorem  \ref{thm:general_thm}}
	\label{sec:proofthm1}
	This section is devoted to the proof of Theorem \ref{thm:general_thm}, which is quite involved and requires some preliminary groundwork.
	\\ In the first part of the proof we lay the basis for the main argument by rephrasing some known results from \cite{WZ99,DM23}, obtaining in particular a functional that represents $\succeq$, which turns out to be additive on measurable sets. We provide some auxiliary lemmata: Lemma \ref{claimpropertiesuplus} proves that it is possible to construct an integrand function with the desired properties; this is based on Lemma \ref{lem:AqBq} which firstly ensures the suitable measurability by a construction on rational numbers; Lemma \ref{claim:pointwise_limit} states that the integral representation is well defined and enjoys the desired properties; finally Lemma \ref{claimvisint} shows continuity of $\mbu$. The concluding argument leads to Proposition \ref{lem:from_v_to_up} and illustrates how ultimately Theorem \ref{thm:general_thm} follows. 
	
	\medskip
	Consider a preference order $\succeq$ on $\brv$ satisfying (SM), (PC) and (ST). Suppose additionally that $\Fcal$ contains at least three disjoint events $A_1, A_2, A_3$ such that $A_i \notin \Ncal_\succeq$ for $i=1,2,3$. As proved in Proposition \ref{prop:exist_unique} (see also \cite{WZ99}, Theorem 11), under these assumptions $\succeq$ admits a numerical representation, i.e. there exists a functional $T: \Bd \to \R$ such that for any given $f, g \in \Bd$ 
	\begin{equation*}
		T(f) \geq T(g) \quad \text{ if and only if } \quad f \succeq g.
	\end{equation*}
	In particular \cite{WZ99}, Theorem 11, or \cite{CL06}, Theorem 4, imply that the Assumptions in \cite{DM23}, Theorem 5.4\footnote{Adopting $\Fcal$ in place of $\Gcal$}, are met by the functional $T$ and we can therefore provide a further refinement of the representing functional. Specifically one can find a $V:\Fcal \times \brv \to \R, \; (A, f) \mapsto V_A(f)$ such that for $f_1, f_2 \in \brv$ and $A \in \Fcal$ we have $f_2 \ind_A \succeq f_1 \ind_A$ if and only if $V_A(f_1) \leq V_A(f_2)$. 
	Furthermore, $V$ can be taken to satisfy the following properties (we refer the reader to the aforementioned literature for further details): 
	\begin{enumerate}[(C1)]
		\item \label{item:signed_measure} for every $f \in \Bd$ the map $A \mapsto V_A(f)$ is a signed measure on $(\Om, \Fcal)$ with $V_A(0) = 0$ for every $A \in \Fcal$. Moreover, $A \mapsto V_A(\ind_\Om) =: \PW(A)$ is a probability measure on $(\Om, \Fcal)$;
		\item \label{item:ind_A} for every $f \in \Bd$, $A \in \Fcal$ we have $V_A(f) = V_{\Omega}(f \ind_A)$;
		\item \label{item:Vmon} the functional $V$ is strictly monotone in that for any $f,g\in \brv$ such that $f\leq g$ then $V_A(f)\leq V_A(g)$ for any $A\in \Fcal$, and if $A$ is such that $\PW(A) > 0 $ then $V_A(x)<V_A(y)$ for any $x,y\in\R$ with $x<y$;
		\item \label{item:PC} for every $A \in \Fcal$ the functional $f \mapsto V_A(f)$, defined on $\Bd$, is pointwise continuous, in that for any sequence $(f_n)_n\subset \brv$, $\sup_n\norm{f_n}_\infty<+\infty$, $f_n(\omega)\to f(\omega)$ for any $\omega\in\Omega$, then $V_A(f_n)\to V_A(f)$.
	\end{enumerate}
	
	The following proposition provides a general representation result which stretches somewhat beyond the scope of the present work. A similar problem is in fact addressed within some relevant contributions in the stream of literature of \cite{AZ90} and \cite[Ch. 2]{Buttazzo}, under the name of Nonlinear Superposition Operator and exploiting the theory of Carath\'eodory functions. Another application of these representation results to subjective expected utility is presented in \cite{STANCA20}.
	
	\begin{proposition} \label{lem:from_v_to_up}
		Let $V:\Fcal \times \brv \to \R$ be a functional satisfying (C\ref{item:signed_measure}), (C\ref{item:ind_A}), (C\ref{item:Vmon})and (C\ref{item:PC}). Then there exists a $\Fcal$-regular function $\mbu:\Om \times \R \to \R$ with $\mbu(\cdot, x) \in \Lcal^1(\Om, \Fcal, \PW)$ for every $x \in \R$,
		such that
		\begin{equation*}
			V_A(f) = \int_A \mbu(\omega, f(\omega)) \dP \quad \text{for any } f \in \brv, A\in\Fcal.
		\end{equation*}
	\end{proposition}
	
	
	Consider a functional $V: \Fcal \times \Bd \to \R$ satisfying (C\ref{item:signed_measure}), (C\ref{item:ind_A}), (C\ref{item:Vmon}) and (C\ref{item:PC}). For any $A \in \Fcal, f \in \Bd$ we denote
	$$
	\mu_f(A) := V_{\Om \cap A}(f) = V_A(f),
	$$
	and observe that, for all $f \in \brv$, $\mu_f$ is dominated by the probability $\PW$ defined as in (C\ref{item:signed_measure}). 
	
	\medskip
	
	We now construct a special version of the state-dependent utility, which enjoys the desired properties.
	
	\begin{definition}\label{def:radon_nykodim}
		For any $q \in \PWq$, we let $u(\cdot, q)\in  \Lcal^1(\Om, \Fcal, \PW)$ be a version of the Radon-Nikodym derivative $\frac{\mathrm{d}\mu_q}{\mathrm{d}\PW}$. In particular, we set $u(\om, 0) = 0$ for all $\om \in \Om$ by virtue of (C\ref{item:signed_measure}). 
	\end{definition}
	For any $A \in \Fcal$ and $q \in \QW$, Definition \ref{def:radon_nykodim} motivates the following notation
	\begin{equation}\label{derivonrationals}
		V_A(q) = \int_A \mathrm{d}\mu_q(\om) =  \int_\Om u(\omega,q) \ind_A(\omega) \dP = \int_\Om u(\omega,q\ind_A(\omega))  \dP   
	\end{equation}
	
	\begin{lemma}\label{lem:AqBq}
		Denote by 
		$$
		A^\PWq := \bigl\{\omega \in \Om: u(\omega, q_1) < u(\omega, q_2), \; \forall \; q_1 <  q_2 \in \PWq \bigr\}$$ and 
		$$
		B^\PWq := \bigl\{ \omega \in \Om: u(\omega, q) = \inf_{\PWq \ni \tilde{q} > q} u(\omega, \tilde{q}), \; \forall \; q \in \PWq\bigl\}.
		$$ 
		Then, $A^\PWq, B^\PWq \in \Fcal$ and $\PW(A^\PWq) = \PW(B^\PWq) = 1$.
	\end{lemma}
	\begin{proof}
		We start fixing $q_1 < q_2 \in \PWq$ and denoting by $C=\{\omega\in\Omega\mid u(\omega, q_1) \geq u(\omega, q_2)\}$. By \eqref{derivonrationals} we can write 
		\begin{equation*}
			V_C(q_1) = \int_C u(\om, q_1) \dP \geq \int_C u(\om, q_2) \dP = V_C(q_2),
		\end{equation*}
		which contradicts (C\ref{item:Vmon}) unless $\PW(C)= 0$. We conclude that $\PW(\{\omega \in \Om: u(\omega, q_1) < u(\omega, q_2)\}) = 1$. 
		By (countable) intersection over the rational numbers we obtain
		\begin{equation*}
			A^\PWq = \bigcap_{\substack{q_1, q_2 \in \PWq,\\ q_1 < q_2}} \bigl\{\omega \in \Om: u(\omega, q_1) < u(\omega, q_2)\bigr\}
		\end{equation*}
		and clearly $A^\PWq \in \Fcal$, with $\PW(A^\PWq)=1$.
		\\ We now turn our attention to $B^\PWq$. We show that 
		\begin{equation}\label{eq:less_or_eq}
			\PW(\{\omega \in \Omega: u(\omega, q) \leq \inf_{\tilde{q} > q} u(\omega, \tilde{q}), \; \forall \; q \in \PWq\}) = 1
		\end{equation} and that 
		\begin{equation}\label{eq:greater_or_eq}
			\PW(\{\omega \in \Omega: u(\omega, q) \geq \inf_{\tilde{q} > q} u(\omega, \tilde{q}), \; \forall \; q \in \PWq\}) = 1,
		\end{equation} so that \eqref{eq:less_or_eq} and \eqref{eq:greater_or_eq} together imply our claim. 
		First of all we show that \eqref{eq:less_or_eq} holds. Fix $q, \tilde{q}\in \PWq$ such that $\tilde{q} > q$, the previous argument shows that the event $\{\omega \in \Omega: u(\omega, q) \leq u(\omega, \tilde{q})\}$ has probability $1$. Notice that for $q \in \PWq$ fixed, we have 
		$$
		\bigcap_{\tilde{q} > q} \{\omega \in \Omega: u(\omega, q) \leq u(\omega, \tilde{q})\} = \{\om \in \Om: u(\omega, q) \leq \inf_{\tilde{q} > q} u(\omega, \tilde{q})\}.
		$$ We denote by $E_q$ the event $\bigcap_{\tilde{q} > q} \{\omega \in \Omega: u(\omega, q) \leq u(\omega, \tilde{q})\}$. Naturally, the countable intersection is still measurable and in particular $\PW(E_q)=1$. Finally, since 
		$$
		\bigcap_{q \in \PWq}E_q = \{\omega \in \Omega: u(\omega, q) \leq \inf_{\tilde{q} > q} u(\omega, \tilde{q}), \; \forall \; q \in \PWq\},
		$$
		we find $\PW\left(\cap_{q \in \PWq}E_q\right) = 1$ and \eqref{eq:less_or_eq} holds.
		\\We proceed proving \eqref{eq:greater_or_eq} by contradiction and set 
		$$A= \{\omega \in \Omega: u(\omega, q) < \inf_{\tilde{q} > q} u(\omega, \tilde{q})\}.$$ 
		Suppose that $\PW\left(A\right) > 0$, then, for some $n \in \mathbb{N}$, $A_n := \{\omega \in \Omega: u(\omega, q) + \frac{1}{n} < \inf_{\tilde{q} > q} u(\omega, \tilde{q})\}$ is such that $\PW(A_n) >0$. Then it holds that 
		\begin{equation}\label{eq:int_onBn}
			\int_{A_n}\left(u(\omega, q) + \frac{1}{n}\right)\dP \leq \int_{A_n} \inf_{\tilde{q} > q} u(\omega, \tilde{q})\dP \leq \inf_{\tilde{q} > q} \int_{A_n} u(\omega, \tilde{q})\dP.
		\end{equation}
		It follows from \eqref{eq:int_onBn} that 
		\begin{equation}
			V_{A_n}(q) + \frac{1}{n}\PW(A_n) \leq \inf_{\tilde{q} > q}V_{A_n}(\tilde{q}).
		\end{equation}
		Since by (C\ref{item:Vmon}) and (C\ref{item:PC}) $V$ is strictly monotone and pointwise continuous we have that $\inf_{\tilde{q} > q}V_{A_n}(\tilde{q}) = V_{A_n}(q)$, which in turn yields the desired contradiction. To conclude, since both \eqref{eq:less_or_eq} and \eqref{eq:greater_or_eq} hold with probability $1$, we have 
		$$\PW(B^\PWq) = \PW(\{\omega \in \Omega: u(\omega, q) = \inf_{\tilde{q} > q} u(\omega, \tilde{q}), \; \forall \; q \in \PWq\}) = 1.$$ 
	\end{proof}
	
	We define the set 
	\begin{align}\label{eq:theta_set}
		\Theta = A^\PWq \cap B^\PWq\in \Fcal \text{ with } \PW(\Theta) = 1,
	\end{align}
	and recall that for any $\omega\in\Theta$ we have $u(\omega, q_1) < u(\omega, q_2)$ for all $q_1 < q_2$, $q_1,q_2\in \PWq$ and  $u(\omega, q) = \inf_{\tilde{q} > q} u(\omega, \tilde{q})$ for every $q\in \PWq$. 
	
	\begin{definition}
		\label{defu+}
		For any $x \in \R$ and $\omega \in \Omega$ we define 
		\begin{equation*}
			u^+(\omega, x) = \inf_{\substack{q \in\PWq,\\ q\geq x}} \{u(\omega, q) \ind_\Theta(\omega) + x \ind_{\Om \setminus \Theta}(\omega)\}.
		\end{equation*} 
	\end{definition}
	We observe that $u^+(\omega, x)$ is $\Fcal$-measurable for any $x \in \R$, as it is the countable pointwise infimum of measurable functions. Notice that $u^+(\om, 0) = 0$ for all $\om \in \Om$.
	
	\begin{lemma}\label{claimpropertiesuplus} Let $u^+:\Omega \times \R \to \R$ be as in Definition \ref{defu+}. Then the following hold:
		\begin{enumerate}[(i)]
			\item \label{item:rationals} $u^+(\omega, q) = u(\omega, q)$ for any $\omega \in \Theta$ (as defined in \eqref{eq:theta_set}) and $q \in \PWq$;
			\item \label{item:mon} if $x < y$, $u^+(\omega, x) < u^+(\omega, y)$ for every $\omega \in \Om$;
			\item \label{item:integrability} $u^+(\cdot, x) \in \Lcal^1(\Om, \Fcal, \PW)$ for every $x \in \R$;
			\item \label{item:RCLL} for every $\omega \in \Om$ the function $u^+(\omega, \cdot)$ is RCLL\footnote{namely right continuous with left limits} on $\R$;
		\end{enumerate}
	\end{lemma}
	
	\begin{proof}
		Item (\ref{item:rationals}) follows from the definition of $\Theta$. 
		Item (\ref{item:mon}) holds trivially on $\Om \setminus \Theta$. Take now any $\omega \in \Theta$: provided that $x<y$, there exist $q_1, q_2 \in \PWq$ such that $x \leq q_1 < q_2 \leq y$. By definition of $\Theta$ we have $u(\omega, q_1) < u(\omega, q_2)$ and moreover by monotonicity of $u$ and definition of $u^+$
		\[u(\omega,q)\leq u^+(\omega,y),\quad \text{for every }q\in\PWq,y\in\R\text{ s.t. }q\leq y.\] Then, we see
		\begin{equation*}
			u^+(\omega, x) = \inf_{\substack{q \in \PWq, \\ q \geq x}} u(\omega, q) \leq u(\omega, q_1) < u(\omega, q_2) \leq \inf_{\substack{q \in \PWq, \\ q \geq y}} u(\omega, q) = u^+(\omega, y).
		\end{equation*}
		As to item (\ref{item:integrability}), fix any $x \in \R$ and consider $q_1, q_2 \in \PWq$ such that $q_1 \leq x \leq q_2$. Then by definition and item (\ref{item:rationals}) $u(\omega, q_1) \leq u^+(\omega, x) \leq u(\omega, q_2)$ for any $\omega \in \Theta$. Since $\PW(\Theta) = 1$ and $u(\cdot, q_1), u(\cdot, q_2)$ are versions of the Radon-Nikodym derivatives $\frac{d\mu_{q_1}}{\mathrm{d}\PW}, \frac{d\mu_{q_2}}{\mathrm{d}\PW} \in \Lcal^1(\Om, \Fcal, \PW)$, integrability follows. 
		
		Finally, observe that item (\ref{item:RCLL}) holds trivially for $\omega \in \Omega \backslash \Theta$. Fix now $\omega \in \Theta$, we wish to show that for any $\varepsilon >0$ there exists $\delta$ such that $|u^+(\omega,x) - u^+(\omega,y)| < \varepsilon$ for any $y \in (x, x+\delta)$. Fix $\varepsilon > 0 $ and let $(q_n)_{n} \subset \PWq$ be such that $q_n \downarrow_n x$. Then there exists $\bar{n} \in \N$ such that, for all $ n > \bar{n}$ we have $|u(\omega, q_n) - \inf_{q > x} u(\omega, q)|  = u(\omega, q_n) - \inf_{q > x} u(\omega, q) < \varepsilon $. Recalling that, from item (\ref{item:rationals}), $u(\omega, q_n) =  u^+(\omega, q_n)$ and $\inf_{q > x} u(\omega, q) = u^+(\omega, x)$ by definition, monotonicity of $u^+$ implies that any $y \in (x, q_{\bar{n}}) \subset \R$ satisfies $|u^+(\omega, x) - u^+(\omega, y)| < \varepsilon$. Existence of the left limit follows again from monotonicity proved in item (\ref{item:mon}).
	\end{proof}
	
	\begin{lemma}\label{claim:pointwise_limit}
		Let $u^+:\Omega \times \R \to \R$ be as in Definition \ref{defu+}. Then the map 
		\begin{equation}
			f \mapsto T_{u^+}(f) := \int_{\Om} u^+(\omega, f(\omega)) \dP
		\end{equation}
		is well defined, monotone with respect to the pointwise order in $\brv$, pointwise continuous\footnote{As previously mentioned: for any uniformly bounded sequence $(f_n)_n\subset \brv$, $f_n(\omega)\to f(\omega)$ for any $\omega\in\Omega$, then $T_{u^+}(f_n)\to T_{u^+}(f)$} and it satisfies $V_{A}(f) = T_{u^+}(f\ind_A)$ for every $f \in \Bd$ and $A\in \Fcal$.
	\end{lemma}
	\begin{proof}
		Observe that, since $u^+(\omega, \cdot)$ is RCLL for every $\omega \in \Om$, by Lemma \ref{lemma:wellposed} we have that $T_{u^+}$ is well defined on $\Bd$ and monotone as $f(\omega)\leq g(\omega)$ implies $u^+(\omega, f(\omega))\leq u^+(\omega, g(\omega))$. We now show that $V_\Om(f) = T_{u^+}(f)$ for every $f \in \Bd$, then pointwise continuity of $T_{u^+}$ will follow from  (C\ref{item:PC}). Fix any $f \in \Bd$, we can find two sequences of simple functions $\underline{f_n}=\sum_{A\in\underline{\pi}_n}\underline{x}^n_A\ind_A$, $\overline{f_n}=\sum_{B\in\overline{\pi}_n}\overline{x}^n_B\ind_B$, with $\overline{\pi}_n,\underline{\pi}_n\subseteq \Fcal$ finite partitions of $\Omega$,  such that $\underline{f_n} \uparrow_n f, \overline{f_n} \downarrow_n f$. Moreover by standard arguments we can assume for any $n$ that $\underline{x}^n_A,\overline{x}^n_B\in \mathbb{Q}$ for any $A\in\underline{\pi}_n$ and $B\in\overline{\pi}_n$. We have 
		as a consequence of \eqref{derivonrationals} and Lemma \ref{claimpropertiesuplus}
		\begin{align*}
			V_{\Omega}(\underline{f}_n)&= \sum_{A\in \underline{\pi}_n}\int_\Omega u(\omega,\underline{x}_n^A)\ind_A\dP = \sum_{A\in \underline{\pi}_n}\int_\Omega u^+(\omega,\underline{x}_n^A)\ind_A\dP
			\\&= \int_\Omega u^+(\omega,\underline{f}_n(\omega))\dP
			\leq\int_\Omega u^+(\omega,f(\omega))\dP
		\end{align*}
		and by a similar argument we can get
		\[\int_\Omega u^+(\omega,f(\omega))\dP\leq \int_\Omega u(\omega,\overline{f}_n(\omega))\dP = V_{\Omega}(\overline{f}_n)\]
		
		By (C\ref{item:PC}) we get the desired equality $V_\Omega(f)=T_{u^+}(f)$ for every $f\in\brv$. Finally, we see that (C\ref{item:ind_A}) yields  $V_A(f)=V_\Omega(f\ind_A)=T_{u^+}(f\ind_A)$, which concludes the proof. 
	\end{proof}
	\begin{lemma}
		\label{claimvisint}
		There exists a function $\mbu: \Om \times \R \to \R$ such that $\mbu(\omega, \cdot)$ is continuous and strictly increasing on $\R$ for every $\omega \in \Om$, $\mbu(\cdot, x) \in \Lcal^1(\Om, \Fcal, \PW)$ for every $x \in \R$ and
		\begin{equation}
			V_A(f) = \int_{A} \mbu(\omega, f(\omega)) \dP \quad \forall \; f \in \Bd \text{ and } A\in\Fcal.
		\end{equation}
	\end{lemma}
	\begin{proof}
		We observe that, by Lemma \ref{claimpropertiesuplus}, $u^+(\omega, \cdot)$ is strictly increasing, RCLL on $\R$ for all $\om \in \Om$ and $u^+(\cdot, x) \in \Lcal^1(\Om, \Fcal, \PW)$ for every $x \in \R$.  Moreover, as argued in Lemma \ref{claim:pointwise_limit}, $T_{u^+}$ is pointwise continuous and in particular continuous from below. Consequently the assumptions of Corollary \ref{corexistscontversion} are verified and its application yields the function $\widehat{u^+}$ which we rename $\mbu := \widehat{u^+}$. Finally, $\mbu$ satisfies all the requirements since
		\begin{equation*}
			T_{\mbu}\stackrel{\text{Def.}}{=}T_{\widehat{u^+}}\stackrel{\text{Cor.}\ref{corexistscontversion}}{=}T_{u^+}\stackrel{\text{Lem.} \ref{claim:pointwise_limit}}{=}V.
		\end{equation*}
	\end{proof}
	
	\noindent\textbf{Conclusions.} We conclude the section by proving Proposition \ref{lem:from_v_to_up} and we show that Theorem \ref{thm:general_thm} follows. 
	
	Let $V:\Fcal \times \brv \to \R$ be a functional satisfying (C\ref{item:signed_measure}), (C\ref{item:ind_A}), (C\ref{item:Vmon}) and (C\ref{item:PC}). Lemma \ref{lem:AqBq} and \ref{claimpropertiesuplus} provide the construction of a suitable integrand function $u^+$ and Lemma \ref{claim:pointwise_limit} ensures that the map $f \mapsto \int_{\Om} u^+(\omega, f(\omega)) \dP$ is well defined. Lemma \ref{claimvisint} concludes the proof by showing the existence of the desired version $\mbu$ with the required properties.
	
	To prove Theorem \ref{thm:general_thm}, we exploit Proposition \ref{lem:from_v_to_up}. Consider a preference $\succeq$ that is (SM), (PC) and (ST) and assume that $\Fcal$ contains at least three disjoint events $A_1, A_2, A_3$ such that $A_i \notin \Ncal_{\succeq}$ for $ i=1,2,3$.  As argued in the beginning of Section \ref{sec:proofthm1}, the preference relation $\succeq$ induces a functional  
	$V:\Fcal \times \brv \to \R$ that satisfies (C\ref{item:signed_measure}), (C\ref{item:ind_A}), (C\ref{item:Vmon}) and (C\ref{item:PC})\, with $V_\Omega(\cdot)$ acting as numerical representation of the preference $\succeq$. Application of Proposition \ref{lem:from_v_to_up} provides a pair $(\mbu, \PW)$ such that $V_{\Omega}(f) = \int_\Om \mbu(\omega, f(\omega)) \dP$ for any $f \in \brv$. 
	
	
	\subsection{Proof of Theorem \ref{CCE}}\label{proof:CCE}
	We begin by recalling that we work under the hypotheses of Theorem \ref{thm:general_thm}. In analogy to the previous section, it is possible to construct a functional $V:\Fcal \times \brv \to \R$ that satisfies (C\ref{item:signed_measure}), (C\ref{item:ind_A}), (C\ref{item:Vmon}) and (C\ref{item:PC}). We also point out that properties (C1-4) are in fact stronger than those required in \cite{DM23}, Theorem 3.2\footnote{Indeed, using now the properties defined in \cite{DM23}, the fact that $V_{\Omega}:\brv\to \R$ is ($\Gcal$-Mo), ($\Gcal$-PC) and ($\Gcal$-NB) follows trivially from (C\ref{item:Vmon}) and (C\ref{item:PC}), for any sub-$\sigma$-algebra $\Gcal$ of $\Fcal$. Furthermore, recalling that, by (C\ref{item:signed_measure}), $V_A(0) = 0$ we are able to recover also ($\Gcal$-QL) and ($\Gcal$-PS). Notice that, for $g_1, g_2 \in \Lcal^\infty(\Om, \Gcal), A \in \Gcal$, it is possible to write $V_{\Omega}(g_1\ind_{A} + g_2 \ind_{A^c}) = V_A(g_1) + V_{A^c}(g_2)$.}. Therefore, $\cm{f \st \Gcal}\neq \emptyset$ and for any $g_1,g_2\in \cm{f\st\Gcal}$ we have $\{g_1\neq g_2\}\in\Ncal_\succeq$.
	
	\medskip
	
	\noindent \textbf{Step 1:} We first show the existence of a state-dependent utility on $\brvG$ for any $\sigma$-algebra $\Gcal\subseteq \Fcal$. 
	Consider the functional $V:\Fcal \times \brv \to \R$ defined above. By Proposition \ref{lem:from_v_to_up} it follows that $V$ admits a representation of the kind $V_A(f) = \int_A \mbu(\om, f(\om))\dP$. Let $\Gcal \subseteq \Fcal$ be an arbitrary sub-\sigmalg on $\Om$ and denote by $V^\Gcal: \Gcal \times \brvG \to \R$ the restriction of the functional $V$ to $\Gcal \subseteq \Fcal$ and $\brvG$. We stress that, differently from $\Fcal$, the \sigmalg $\Gcal$ need not contain at least three disjoint events of positive probability.  Analogously to $V$, $V^\Gcal$ defines a probability measure  $\PW_\Gcal$ on $\Gcal$, which is the restriction of the probability measure $\PW$ induced by $V$ to the sub-\sigmalg $\Gcal$. We will hence use only $\PW$ in the following, to avoid impractical notations. 
	\\ Applying Proposition \ref{lem:from_v_to_up} to $V^\Gcal$ it is possible to construct a $\Gcal \otimes \Borel_\R$-measurable function $\ug: \Om \times \R \to \R$ such that
	\begin{enumerate}[(i)]
		\item $\ug(\om, \cdot)$ is continuous and strictly increasing on $\R$ for all $\om \in \Om$,
		\item  $\ug(\cdot, x) \in \Lcal^1(\Om, \Gcal, \PW)$ for all $x \in \R$,
		\item  $V_A^\Gcal(g) = \int_A \ug(\om, g(\om)) \dP, \, \forall A \in \Gcal, g \in \brvG$.
	\end{enumerate}
	Notice that for all  $A \in \Gcal$ and  $g \in \brvG$
	\begin{equation*}
		V_A^\Gcal(g) = \int_A \ug(\om, g(\om)) \dP = \int_A \mathbb{u}(\om, g(\om)) \dP = V_A(g).
	\end{equation*}
	It is then immediate to verify by definition of conditional expectation that the random variable $\ug(\cdot, x)$ is an element of the class $\Ep{\mathbb{u}(\cdot, x)|\Gcal}$. 
	
	\medskip
	
	To avoid overburdening notation, in the following step 2 we shall write $\Ep{\mbu(f)\st\Gcal}$ in place of $\Ep{\mathbb{u}(\cdot, f(\cdot))|\Gcal}$ and $\ug(g)$ instead of $\ug(\cdot, g(\cdot))$. Moreover, we shall often make the following standard abuse of notation: for any $f\in\brv$ we write $g=\Ep{\mathbb{u}(f)|\Gcal}$ to indicate  that $g \in \Ep{\mathbb{u}(f)|\Gcal}$ i.e. $g$ is a version of the conditional expectation $\Ep{\mathbb{u}(f)|\Gcal}$.
	
	\medskip
	
	\noindent \textbf{Step 2:} We show that:
	\begin{equation}
		\label{almosttakeout}
		\ug(g) \text{ belongs to }\Ep{\mbu(g)\middle|\Gcal}\quad \forall \,g\in\brvG.
	\end{equation}
	Consider indeed a simple function $g = \sum_{A \in \pi}y_A \ind_A$, with $\pi$ being a finite $\Gcal$-measurable partition of $\Om$. Then we have
	\begin{align*}
		\Ep{\mbu(g)\middle|\Gcal}
		&=\sum_{A \in \pi} \Ep{\mbu(g)\ind_{A}\middle|\Gcal} 
		= \sum_{A \in \pi}\Ep{\mbu(y_A)\ind_{A}\middle|\Gcal}\\
		&=\sum_{A \in \pi}\Ep{\mbu(y_A)\middle|\Gcal}\ind_{A}
		=\sum_{A \in \pi} \ug(y_A)\ind_{A}  =  \ug(\sum_{A \in \pi} y_A \ind_{A} ),
	\end{align*}
	and therefore \eqref{almosttakeout} is verified for $\Gcal$-measurable simple functions. We proceed by showing that the property also holds for general functions $g \in \brvG$. Take any sequence of $\Gcal$-measurable simple functions $(\underline{g}_n)_n \in \brvG$ such that $\underline{g}_n \uparrow_n g$ (again, the existence of such functions is guaranteed by standard arguments). Monotonicity of $\ug(\om, \cdot)$ ensures that $(\ug(\omega, \underline{g}_n(\om)))_n$ is an increasing sequence for all $\om \in \Om$. Finally, pointwise continuity of $\mathbb{u}$ gives
	\begin{equation*}
		\ug(g)=\lim_n \ug(\underline{g}_n),
	\end{equation*}
	and monotone convergence theorem yields \eqref{almosttakeout} for general $g \in \brvG$ as
	\begin{equation*}
		\lim_n \Ep{\mbu(\underline{g}_n) \st \Gcal} = \Ep{\mbu(g)\middle|\Gcal}.
	\end{equation*}

	\medskip
	
	\noindent\textbf{Step 3:} we now study the generalized inverse of the state-dependent utility $\ug$. If on the one hand  $(\omega,x)\mapsto \ug(\omega,x)$ is $\Gcal \otimes \Borel_\R$ measurable (by Proposition \ref{lem:from_v_to_up}), on the other we would like the generalized inverse $\Phi_\Gcal$ to enjoy the same measurability as $\ug$. We first extend the domain of  $\ug:\Om \times \R \to \R$ to $\Om \times \overline{\R}$ by denoting
	\begin{equation*}
		\ug(\om, +\infty):= \lim_{x\to+\infty}\ug(\om, x), \quad\quad \ug(\om, -\infty):= \lim_{x\to-\infty}\ug(\om, x).
	\end{equation*}
	Monotonicity and continuity of $\ug(\om, \cdot)$ for all $\om \in \Om$ ensure that the limits exist. Fixed $\om \in \Om$, we denote with $\ima{\ug} := \text{Im}(\ug(\om, \cdot))$ the image of $\ug$, that is the interval $(\ug(\om, -\infty), \ug(\om, +\infty))$, and we stress that it is $\om$-dependent and open by strict monotonicity. \\
	Let $\Phi_\Gcal:\Om \times \R \to \overline{\R}$ be the $\om$-wise right-continuous generalized inverse:
	$$
	\Phi_\Gcal(\omega,x):=\inf\{y\in\R\mid \ug(\omega,y)>x\}.
	$$
	Observe that strict monotonicity and continuity of $\ug(\om, \cdot)$ imply  $\Phi_\Gcal(\om, \cdot)$ is strictly increasing and continuous over $\ima{\ug}$ for all $\om \in \Om$. Then, in particular, $\ug(\om, \cdot): \R \to \ima{\ug}$ is a bijection and we can write 
	\begin{equation}\label{eq:phi_explicit}
		\Phi_\Gcal(\om, x) = 
		\begin{cases}
			+\infty &\quad \text{for } x \geq \ug(\om, +\infty),\\
			\ug^{-1}(\om, x) &\quad \text{for } x \in \ima{\ug}, \\
			-\infty &\quad \text{for } x \leq \ug(\om, -\infty).
		\end{cases}
	\end{equation}
	Notice that, fixed $\om \in \Om$, whenever $x \in \ima{\ug}$ then $\Phi_\Gcal(\om, x)$ is finite.
	
	We claim that $\Phi_\Gcal(\cdot, x): \Om \to \overline{\R}$ is $\Gcal$-measurable for all $x \in \R$. To see this, let $x \in \R$ be fixed and consider the set $A = (a, +\infty] \in \Borel_{\overline{\R}}$ with $a > -\infty$. We show that $\left\{\om \in \Om : \Phi_\Gcal(\om, x) \in A \right\} \in \Gcal$.
	We observe that, by strict monotonicity of $\ug$ the latter coincides with $\left\{\om \in \Om : x > \ug(\om, a) \right\}$ which is an element of $\Gcal$ as $\ug$ is $\Gcal$-measurable.
	\\ Additionally, we wish to emphasize that $\Phi_\Gcal(\om, \cdot): \R \to \extR$ is continuous for all $\om \in \Om$, and therefore $\Phi_\Gcal:\Om \times \R \to \extR$ is a Carath\'eodory function (see \cite{Aliprantis}, Definition 4.50). In particular, we observe that $\R$ with the usual euclidean topology is a separable metrizable space and $\overline{\R} = \extR$ with the standard topology on the extended real numbers is metrizable.
	Then, by virtue of \cite{Aliprantis}, Lemma 4.51, $\Phi_\Gcal$ is $\Gcal \otimes \Borel_\R$ measurable.
	
	\begin{remark}
		By \eqref{eq:phi_explicit}, for any $\om \in \Om$ and  $x \in \mathrm{Im}_\om(\ug)$ we have 
		\begin{equation}\label{inverseGmeas}
			\ug\Big(\om, \Phi_\Gcal(\om, x)\Big) = x.
		\end{equation}
	\end{remark}
	
	\noindent\textbf{Step 4:} we conclude proving the desired representation of the conditional Chisini mean. Chosen $f\in \brv$,  by strict monotonicity of $\mbu(\om, \cdot)$ we have that 
	\begin{equation*}
		\mbu(\om, -\snorm{f}) \leq \mbu(\om, f(\om)) \leq \mbu(\om, -\snorm{f}), \quad \quad \forall \; \om \in \Om.
	\end{equation*}
	Taking a version $g$ of $\Ep{\mbu(\cdot, f)\middle|\Gcal}$ we can find $A \in \Gcal$ such that $\PW(A) = 1$ and exploiting \eqref{almosttakeout} we may write
	\begin{equation}\label{eq:range_g}
		\ug(\om, -\snorm{f}) < g(\om) < \ug(\om, \snorm{f}) \quad \text{ for all } \om \in A.
	\end{equation}
	Let us define $\tilde{g} = g\ind_{A} + 0 \ind_{A^c}$ and observe that now $\tilde{g}(\omega) \in \ima{\ug}$ for all $\om \in \Om$. Moreover, $\tilde{g}$  is still a version of the conditional expectation of $f$. Recalling that the mapping $\Phi_\Gcal(\om, \cdot): \R \to \overline{\R}$ is strictly increasing on $\ima{\ug} \subseteq \R$ for all $\om \in \Om$, an application to  \eqref{eq:range_g} yields 
	$$
	\Phi_\Gcal(\om, \ug(\om, -\snorm{f})) < \Phi_\Gcal(\om, \tilde{g}(\om)) < \Phi_\Gcal(\om, \ug(\om, \snorm{f})).
	$$
	Consequently, by \eqref{inverseGmeas} $\Phi_\Gcal(\om, \ug(\om, x)) = x$ for all $\om \in \Om$ and all $x \in \R$ and it follows that that $-\snorm{f} < \Phi_\Gcal(\om, \tilde{g}) < \snorm{f}$, for every $\om \in \Omega$. Hence $\Phi_\Gcal\left(\cdot, \tilde{g}\right) \in \brvG$. In particular we have
	\begin{equation}
		\label{uphigisg}
		\ug\Big(\om, \Phi_\Gcal\big(\om, \tilde{g}(\omega)\big)\Big) = \tilde{g}(\omega)\quad\text{ for every }\om\in\Omega.   
	\end{equation}
	
	We prove that $\Phi_\Gcal(\om, \tilde{g}(\omega))$ is the conditional Chisini mean in that it solves the infinite dimensional system of equations given in \eqref{eq:sol_chisinimean},  for the preference representing functional $V_{\Omega}(f)=\int_{\Omega} \mbu(\omega, f(\omega)) \dP$ defined in Theorem \ref{thm:general_thm}. Considering an arbitrary set $A \in \Gcal \subseteq \Fcal$, we have
	\begin{align*}
		V_{\Omega}\big(\Phi_\Gcal(\cdot, \tilde{g})\ind_A\big)&=\Ep{\mbu\big(\cdot, \Phi_\Gcal\left(\cdot,\tilde{g}\right)\ind_A\big)}=\Ep{\Ep{\mbu\big(\cdot, \Phi_\Gcal\left(\cdot, \tilde{g}\right)\big)\ind_A\middle|\Gcal}}\\
		&=\Ep{\Ep{\mbu\big(\cdot, \Phi_\Gcal\left(\cdot, \tilde{g}\right)\big)\middle|\Gcal}\ind_A}
		\stackrel{\eqref{almosttakeout}}{=}\Ep{\ug\big(\cdot,\Phi_\Gcal\left(\cdot, \tilde{g}\right)\big)\ind_A}\\
		&\stackrel{\eqref{inverseGmeas}}{=}\Ep{\tilde{g}\ind_A}=\Ep{\Ep{\mbu(\cdot, f)\ind_A\middle|\Gcal}}=\Ep{\mbu(\cdot, f)\ind_A}\\
		&=\Ep{\mbu(\cdot, f\ind_A)}=V_{\Omega}(f\ind_A),
	\end{align*}
	observing that since $\mbu(\omega,0)=0$ for all $\om \in \Om$ by construction (whence also $\Phi_\Gcal(\om, 0) = 0$ for all $\om \in \Om$) we have 
	\[\begin{cases}
		\mbu(\omega,f(\omega))\ind_A(\omega)=\mbu\big(\omega,f(\omega)\ind_A(\omega)\big)\\\Phi_\Gcal\left(\omega,\tilde{g}(\omega)\right)\ind_A(\omega)=\Phi_\Gcal\big(\omega,\tilde{g}(\omega)\ind_A(\omega) \big)   
	\end{cases}\quad\text{ for every }\om\in\Omega.\]
	
	\subsection{Proof of Corollary \ref{cor:representability}}
	The proof of the corollary is merely an application of the results in the following sections, we briefly outline here the main steps. To show that (\ref{prop:ST}) implies (\ref{prop:CCE}), consider a preference order $\succeq$ that satisfies (SM), (PC) and (ST) and suppose $\Fcal$ contains at least three disjoint non-null events. The existence of a conditional Chisini mean under these assumptions is guaranteed by \cite{DM23}, Theorem 3.2 (See the proof of Proposition \ref{prop:exist_unique} for more details). Representability of conditional Chisini means is then obtained from Theorem \ref{thm:general_thm} and Theorem \ref{CCE} as follows. The former provides a pair ($\mbu, \PW$) such that $\mbu:\Om \times \R \to \R$ is $\Fcal$-regular and $\PW$-integrable. The latter allows the construction for any $\Gcal \subseteq \Fcal$ of a $\Gcal$-regular function $\ug:\Om \times \R \to \R$ that is $\PW$-integrable and such that, for all $f \in \brv$ and for a suitable version $h$ of $\Ep{\mbu(\cdot, f)|\Gcal}$, we have $\ug^{-1}\left(\cdot, h(\cdot)\right) \in \cm{f\middle|\Gcal}$.\\
	Regarding the converse implication, observe that Item (\ref{prop:CCE}) actually implies that any representing functional $T$ for $\succeq$ is conditionable for every choice of $\Gcal \subseteq \Fcal$. Therefore, the same argument used in the proof of Theorem \ref{characterization} (in particular Item (\ref{prop:condA}) $\Rightarrow$ Item (\ref{prop:ST})) shows that $\succeq$ satisfies (ST).

	\subsection{Proof of Theorem \ref{repr:nonlinear}}\label{proof:thm2}
	In this section we deal with equivalence classes of random variables, and therefore for any $f\in \brv$ we denote by $X=[f]_{\PW}$ the element in $\Linf$ such that $f\in X$.
	\medskip
	
	\noindent We first prove the converse implication of the statement. To this aim consider an $\Fcal$-regular function $\mbu:\Om \times \R \to \R$ such that $\mbu(\cdot, x)$ is integrable for all $x \in \R$ and let $\Sigma$ be a collection of sub-$\sigma$-algebras of $\Fcal$. We define a preference relation on $\brv$ by 
	\[f_1\succeq f_2 \text{ if and only if } \int_{\Omega}\mbu(\omega,f_1(\omega))\dP\geq \int_{\Omega}\mbu(\omega,f_2(\omega))\dP. \]
	For any $\Gcal \in \Sigma$ we can directly apply Theorem \ref{CCE} and find
	a $\Gcal$-regular function $\ug:\Om \times \R \to \R$ such that $\ug(\cdot, x)$ is a version of $\Ep{\mbu(x)\middle|\Gcal}$. Now fix $X\in\Linf$ and let $f$ be any representative element of the equivalence class $X$. Analogously to the argument in Section \ref{proof:CCE}, steps 3 and 4, we take a suitable version $h\in \Ep{\mbu(f)\middle| \Gcal}$ and compute
	$\Phi_\Gcal(\cdot, h(\cdot))$ where $\Phi_\Gcal:\Omega\times \R\rightarrow \overline{\R}$ is the right inverse of $\ug$, that is 
	$$\Phi_\Gcal(\omega,x):=\inf\{y\in\R\mid \ug(\omega,y)>x\}.$$
	Observe that $[\Phi_\Gcal(\cdot, h(\cdot))]_{\PW}$ does not depend on the choice of the representative $f\in X$, hence the definition $\ug^{-1}\left(\Ep{\mbu(X)\middle|\Gcal}\right) := [\Phi_\Gcal(\cdot, h(\cdot))]_{\PW}$ is well posed. In particular, the previous argument formalizes the definition of the mapping $\Ecal_\Gcal:X \mapsto \ug^{-1}\left(\Ep{\mbu(X)\middle|\Gcal}\right)$ provided in the statement of Theorem \ref{repr:nonlinear}. It remains to show that $\left\{\Ecal_\Gcal\right\}_{\Gcal \in \Sigma}$ defined above is a family of time consistent conditional nonlinear expectations such that $\Ecal_0$ satisfies Assumption \ref{ass:E0}. It is just a simple checking that for any $X\in\Linf$, $A\in\Gcal$ we have $\Ecal_{\Gcal}(X\ind_A)=\Ecal_{\Gcal}(X)\ind_A$. Similarly, $\Ecal_0(\Ecal_{\Gcal}(X))=\Ecal_0(X)$ for any  $\Gcal\in\Sigma$  and  $X\in\Linf$. Finally we observe that $\Ecal_0(X)=\mbu_0^{-1}\left( \int_{\Omega}\mbu(\omega,f(\omega))\dP\right)$
	for $\mbu_0: x \mapsto \Ep{\mbu(x)}$ and some $f\in X$. Therefore from strict monotonicity and continuity of $\mbu$ we get: 
	(a) for any $x,y\in\R$, $Z\in\Linf$ and $A\in \Fcal$ with $\PW(A)>0$, if $x<y$ then $\Ecal_{0}(x\ind_A+Z\ind_{A^c})<\Ecal_{0}(y\ind_A+Z\ind_{A^c})$; (b) for any bounded sequence $X_n$ converging $\PW$-almost surely to $X$ we have $\lim_{n\to\infty}\Ecal_0(X_n)=\Ecal_0(X)$. This concludes the first implication.


	\noindent Now suppose $\left\{\Ecal_\Gcal\right\}_{\Gcal \in \Sigma}$ is a family of time consistent nonlinear conditional expectations such that $\Ecal_0$ satisfies Assumption \ref{ass:E0}. We show that there exists an $\Fcal$-regular function $\mbu:\Om \times \R \to \R$ such that $\mbu(\cdot,x)$ is integrable for all $x \in \R$ and, for any $X \in \Linf$ and $\Gcal \in \Sigma$, $\Ecal_\Gcal$ is of the form
	$$
	\Ecal_{\Gcal}(X)=\ug^{-1}(\Ep{\mbu(X)\st\Gcal}),
	$$ with $\ug(\cdot, x) = \Ep{\mbu(x)\middle|\Gcal}$. For any $f_1,f_2\in\brv$ we define the preference order $\succeq$ as follows
	\[f_1\succeq f_2 \text{ if and only if } \Ecal_0(X_1)\geq \Ecal_0(X_2),\]
	where $X_1=[f_1]_{\PW},X_2=[f_2]_{\PW} \in\Linf$ are the equivalence classes generated by $f_1,f_2$, and notice that the preference relation $\succeq$ is well defined.
	\\ We start by showing that the null sets of the preference as defined in \eqref{defnulls} correspond to the null sets of the reference probability measure $\PW$ i.e. $\Ncal_{\succeq}=\{A\in\Fcal\mid \PW(A)=0\}$. Let $A\in \Fcal$ be such that $\PW(A)=0$ and consider $f \in \brv$ with $f$ being a representative element in $X \in \Linf$. Take any $g\in\brv$ being a representative of any $Y\in \Linf$. Then since $\PW(A)=0$, we have that $f\ind_{A^c}+g\ind_{A}$ is still a representative of $X$, thus $ \Ecal_0(X\ind_{A^c}+Y\ind_{A}) = \Ecal_0(X)$. This yields $A\in\Ncal_{\succeq}$.
	Now suppose $A \in \Fcal$ is such that $\PW(A) > 0$ and consider $x,y \in \R$ with $x < y$. For any choice of $Z \in \Linf$, monotonicity of $\Ecal_0$ implies that 
	$$
	\Ecal_0(x\ind_A + Z \ind_{A^c}) < \Ecal_0(y\ind_A + Z \ind_{A^c}),
	$$
	which in turn yields $A \notin \Ncal_\succeq$.
	\\ Secondly, we observe that $\succeq$ satisfies (SM) and (PC) as an immediate consequence of Assumption \ref{def:nonlinear}. Finally, notice that for any $f\in\brv$ the functional $T(f)=\Ecal_0(X)$, with $X=[f]_{\PW} \in \Linf$, represents $\succeq$. By the properties of the family $\left\{\Ecal_\Gcal\right\}_{\Gcal \in \Sigma}$, for any sub-$\sigma$-algebra $\Gcal$ we have $\Ecal_0(X\ind_A)=\Ecal_0(\Ecal_{\Gcal}(X)\ind_A)$ so that $T(f\ind_A)=T(g\ind_A)$ for any $A\in \Gcal$ and $g\in \Ecal_{\Gcal}(X)$. Hence $T$ is conditionable and in particular, by Theorem \ref{characterization}, $\succeq$ satisfies (ST).
	By Theorem \ref{thm:general_thm}  we can find a measure $\QW$ and a state-dependent utility $\tilde{\mbu}$ such that Eq. \eqref{repr:functional} holds, adopting the pair $(\tilde{\mbu},\QW)$. We observe that $\QW\sim \PW$ and hence we can  redefine $\mbu(\omega,x)= \frac{d\QW(\omega)}{d\PW}\tilde{\mbu}(\omega,x)$, where $\frac{d\QW}{d\PW}$ is a version of the Radon Nikodym derivative, and observe that Eq. \eqref{repr:functional} still holds true for ($\mbu,\PW$), with $\mbu$ being $\Fcal$-regular.  Applying Theorem \ref{CCE} we indeed can define 
	a $\Gcal$-regular function $\ug:\Om \times \R \to \R$, with $\ug(\cdot, x) \in \Lcal^1(\Om, \Gcal, \PW)$ for every $x \in \R$, as a version of $\Ep{\mbu(x)\middle|\Gcal}$ for every $x\in\R$. In this way, for an appropriate choice of $h\in \Ep{\mbu(f)\middle| \Gcal}$, $\Phi_\Gcal(\cdot, h(\cdot))\in\cm{f\middle|\Gcal}$ where $\Phi_\Gcal:\Omega\times \R\rightarrow \extR$ defined by
	\[\Phi_\Gcal(\omega,x):=\inf\left\{y\in\R \middle| \ug(\omega,y)>x\right\}\]
	is $\Gcal\otimes\Borel_\R$ measurable. 
	Note that the map $X\mapsto [\Phi_\Gcal(\cdot, h(\cdot))]_{\PW}$ is well defined, since the procedure above does not depend (up to $\PW$-null events) on the particular choice of a representative of $X$, and we shall write $\ug^{-1}\left(\Ep{\mbu(X)\middle| \Gcal}\right)=[\Phi_\Gcal(\cdot, h(\cdot))]_{\PW}$. By the previous argument any element $g\in \Ecal_{\Gcal}(X)$ is also an element of $\cm{f\middle|\Gcal}$ and therefore $\ug^{-1}\left(\Ep{\mbu(X)\middle| \Gcal} \right)= \Ecal_{\Gcal}(X)$.
	
	\begin{remark}
		It is worthy of note that the ``taking out what is known'' property in Definition \ref{def:nonlinear} jointly with Assumption \ref{ass:E0} imply that for any $Y \in \LinfG$ we have $\Ecal_\Gcal(Y) = Y$. To see this, first observe that by strict monotonicity and time consistency one can show $\Ecal_\Gcal(x) = x$ for $x \in \R$. Consequently, applying the ``taking out what is known'' property, it follows that $\Ecal_\Gcal(Z) = Z$ for any $\Gcal$-measurable simple function $Z$. To conclude, for a general $Y \in \LinfG$ we can choose a sequence of $\Gcal$-measurable simple functions $(\hat{Y}_N)_N$ such that $\hat{Y}_N \to_N Y$ with respect to the sup norm, then it can be shown that $||\Ecal_\Gcal(Y) - \hat{Y}_N||_\infty \to_N 0$ which in turn yields $\Ecal_\Gcal(Y) = Y$.
	\end{remark}
	
	\appendix
	
	\renewcommand{\thesection}{\Alph{section}}
	
	\section{Appendix: auxiliary results}
	
	\begin{lemma}\label{lem:ptwise_mon}
		Suppose $\succeq$ satisfies (SM) and (PC), then $\succeq$ is monotone on $\brv$, that is, for $f,g \in \brv$, $f(\om) \geq g(\om)$ for all $\om \in \Om$ implies $f \succeq g$.
	\end{lemma}
	
	\begin{proof}
		We begin by showing the result for simple functions. Consider two simple functions $f, g \in \brv$ such that $f(\om) \geq g(\om)$ for all $\om \in \Om$. It is possible to find a common partition $\pi := \left\{A_1, A_2, ..., A_N\right\} \subseteq \Fcal$ for some $N \in \N$ such that we can write $f = \sum_{i = 1}^N x_i \ind_{A_i}$, $g = \sum_{i=1}^N y_i \ind_{A_i}$, with $x_i, y_i \in \R$ for $i=1, ..., N$. We have
		\begin{equation*}
			f = \sum_{i = 1}^N x_i \ind_{A_i}
			\succeq y_1 \ind_{A_1} + \sum_{i = 2}^N x_i \ind_{A_i} 
			\succeq y_1 \ind_{A_1} +  y_2 \ind_{A_2} + \sum_{i = 3}^N x_i \ind_{A_i} \succeq \ldots \succeq \sum_{i = 1}^N y_i \ind_{A_i} = g,
		\end{equation*}
		where $\succeq$ will be $\sim$ if $x_i=y_i$ and $\succ$ if $x_i > y_i$ by (SM). We conclude $f \succeq g$ for any pair of simple functions. \\
		Now consider $f,g \in \brv$ such that $f(\om) \geq g(\om)$ for all $\om \in \Om$. Let $\left(f_n\right)_n, \left(g_n\right)_n \subseteq \brv$ be sequences of simple functions such that $f_n \downarrow_n f$ and $g_n \uparrow_n g$ so that $f_n\geq g_n$ for every $n\in \N$. Therefore $f_n\succeq g_n$ from the previous point, for every $n\in\N$. Now assume by contradiction that $g\succ f$, then applying (PC) we find $n_1$ such that $g_{n_1}\succ f$ and applying again (PC) we find  $n_2$ such that $g_{n_1}\succ f_{n_2}$.  By monotonicity on simple functions we have $g_n\succeq g_{n_1}\succ f_{n_2}\succeq f_n$ for every $n\geq n_1\vee n_2$ and hence a contradiction. 
	\end{proof}
	
	
	\begin{lemma}\label{lemma:wellposed}
		Let $u:\Om \times \R \to \R$ and assume that for every $\omega \in \Om$ the function $x\mapsto u(\omega,x)$ is right continuous with left limit (RCLL), nondecreasing and $u(\omega, 0) = 0$. Suppose additionally that $u(\omega, x) \in \Lcal^1(\Omega, \Fcal, \PW)$ for every $x \in \R$. Then:
		\begin{enumerate}[(i)]
			\item \label{item:joint_measurability} $u$ is $\Fcal \otimes \Borel_\R - \Borel_\R$ measurable;
			\item \label{item:okonf} for every $f \in \Bd$, $\omega \mapsto u(\omega, f(\omega)) \in \Lcal^1(\Om, \Fcal, \PW)$
		\end{enumerate}
	\end{lemma}
	\begin{proof}
		We prove Item (\ref{item:joint_measurability}) in the spirit of \cite{KS91} (Proposition 1.1.13 and Remark 1.1.14) as follows. Let $m > 0$. For $n \geq 1$, $k = 0,1, ..., 2\cdot2^n -1$, we define
		\begin{equation*}
			u^{(n)}(\omega, s) := u\left(\omega, \frac{\left(k + 1\right)m}{2^{n}} - m\right) \;\;\;\;\; \text{ for } \;\;\; k\frac{m}{2^{n}} - m \leq s < \frac{\left(k + 1\right)m}{2^{n}} - m,
		\end{equation*} and fix 
		\begin{equation*}
			\begin{cases}
				u^{(n)}(\omega, s) := u\left(\omega, - m\right), \;\;\;\;\; &\text{ for } s < -m, \\
				u^{(n)}(\omega, s) := u\left(\omega, m\right), \;\;\;\;\; &\text{ for } s \geq m.
			\end{cases}
		\end{equation*} 
		The newly constructed process is piecewise constant over intervals of length $\frac{m}{2^{n}}$ and right continuous. The map $(\omega,s) \mapsto u^{(n)}(\omega, s)$ is $\Fcal \otimes \Borel_\R$-measurable. By right continuity of $u$ we have that $u^{(n)}(\omega, s) \to u(\omega, s)$  for every pair $(\omega, s) \in \Omega \times [-m, m]$ as $n \to \infty$. Therefore we obtain that the map $(\omega,s) \mapsto u(\omega, s)$ is itself $\Fcal \otimes \Borel_{[-m, m]}$-measurable as pointwise limit of a sequence of measurable functions. Finally, by letting $m \to \infty$ we have that $(\omega,s) \mapsto u(\omega, s)$ is $\Fcal \otimes \Borel_\R$-measurable.
		By the latter argument we have that $\omega \mapsto u(\omega, x)$ is $\Fcal \otimes \Borel_\R - \Borel_\R$ measurable. Moreover, for any $\Fcal - \Borel_\R$ measurable and bounded $f$, the map $\varphi: \omega \mapsto (\omega, f(\omega))$ is $\Fcal - \Fcal \otimes \Borel_\R$  measurable (see \cite{Aliprantis}, Lemma 4.49). Hence, the composition $(u \circ \varphi): \omega \mapsto u(\omega, f(\omega))$ is $\Fcal - \Borel_\R$ measurable. To conclude, since $u$ is nondecreasing, $\Lcal^1(\Om, \Fcal, \PW) \ni u(\omega, -\norm{f}_\infty) \leq u(\omega, f(\omega)) \leq u(\omega, \norm{f}_\infty) \in \Lcal^1(\Om, \Fcal, \PW)$ thus concluding the proof of Item (\ref{item:okonf}). 
	\end{proof}
	
	\begin{proposition}
		\label{propcontofuplus}
		Let $u:\Omega\times \R\rightarrow \R$ satisfying
		\begin{enumerate}[(i)]
			\item \label{i} for every $\omega\in \Omega$ $u(\omega,\cdot)$ is RCLL, nondecreasing  and $u(\omega,0)=0$;
			\item \label{ii} $u(\cdot,x)\in \Lcal^1(\Om,\Fcal, \PW)$ for every $x\in\R$;
			\item \label{iii} the functional $T_u: \brv \rightarrow \R$ defined by $T_u(f):=\int_\Omega u(\omega, f(\omega))\dP$ is continuous from below \footnote{For any $(f_n)_{n}\subset \mathcal{L}^\infty(\Omega,\Fcal)$  such that $f_n(\omega)\uparrow_n f(\omega)$ for any $\omega\in\Om$ we have $T_u(f_n)\rightarrow T_u(f)$}.
		\end{enumerate}
		Then $u$ is continuous in the following sense:  $$A_{\text{cont}}:=\{\omega\in \Omega\mid x\mapsto u(\omega,x)\text{ is continuous}\}\in\Fcal$$ and $\PW(A_{\text{cont}})=1$. 
	\end{proposition}

	\begin{proof}
		In the present proof we shall rely on few concepts related to the theory of stochastic processes and in particular on stopping times. We refer the reader to \cite{KS91}, page 6, for an exhaustive treatment of the notions involved.

		Lemma \ref{lemma:wellposed} guarantees that $T_u:\mathcal{L}^\infty(\Omega,\Fcal)\to \R$ is well defined and finite valued.  
		Fix  $M>0$ and define $X_t(\omega):=u(\omega, t-M)$ for $t\geq 0$, omitting the dependence on $M$ for simplicity. Take also the filtration $\Fcal_t:=\Fcal, t\geq 0$, which is trivially right continuous by definition: $\bigcap_{s>t}\Fcal_s=\Fcal=\Fcal_t$. Then $(X_t)_t$ is an RCLL process, adapted to $(\Fcal_t)_t$. Set $\Delta X_0=0$ and for $t>0$ set $\Delta X_t:=X_t-\sup_{0\leq s<t}X_s=X_t-\lim_{s\uparrow t}X_s$. 
		
		Following the idea in \cite{So13}, we define $\tau^\varepsilon$ as the first jump time for the size $\varepsilon>0$, namely $$\tau^\varepsilon:=\inf\{t\geq 0\mid \Delta X_t>\varepsilon\}\stackrel{(\star)}{=}\inf\{t> 0\mid X_t-\sup_{0\leq s<t}X_s>\varepsilon\},$$ 
		where equality $(\star)$ follows from posing $\Delta X_0=0$. Looking back at $u(\omega,\cdot)$, $\tau^\varepsilon(\omega)$ is the  first time $u(\omega,\cdot)$ jumps, in the interval $(-M,+\infty)$, with a jump size strictly greater than $\varepsilon$, with the usual convention that such time is set to be $+\infty$ if no such jump actually occurs. By \cite{So13}, $\tau^\varepsilon:\Omega\rightarrow [0,+\infty]$ is a $(\Fcal_t)_t$-stopping time, hence in particular it is $\Fcal$-measurable by definition of such a filtration. 
		Setting 
		$$A_{cont}^M:=\{\omega\in \Omega\mid [0,+\infty)\ni t\mapsto X_t(\omega)\text{ is continuous}\},$$ 
		we observe that 
		$$\Omega\setminus A_{cont}^M=\bigcup_{n\in\mathbb{N}}\{\tau^{\frac{1}{n}}<+\infty\},$$ 
		hence the measurability of $A^M_{cont}$ follows. Suppose now that $\PW(A^M_{cont})<1$, then for some $n$ we would have $\PW(\tau^{\frac1n}<+\infty)>0$ and by $\{\tau^{\frac1n}<+\infty\}=\bigcup_{N\in \mathbb{N}}\{\tau^{\frac1n}<N\}$ we could find $N$ big enough such that  $\PW(\tau^{\frac1n}\leq N)>0$. 
		\\Set, for a fixed pair  $n,N$ satisfying $\PW(\tau^{\frac1n}\leq N)>0$,  $B:=\{\tau^{\frac1n}\leq N\}$ and $Z:=\tau^{\frac1n} \ind_B\in \mathcal{L}^\infty(\Omega,\Fcal)$.
		For any fixed $\omega\in B$ we know that $\tau^{\frac1n}(\omega)=\inf\{t\geq 0\mid \Delta X_t(\omega)>1/n\}= \lim_{k\to \infty} t_k(\omega)$ for some $\omega$-dependent, non-increasing sequence $t_k(\omega)\downarrow_k\tau^{\frac1n}(\omega)$ which satisfies $\Delta X_{t_k(\omega)}(\omega)>1/n$ for every $k$. This means that, for any $k$,
		\[\frac{1}{n}\leq X_{t_k(\omega)}(\omega)-\sup_{0\leq s<t_k(\omega)}X_s(\omega) \leq X_{t_k(\omega)}-\sup_{0\leq s<\tau^{\frac1n}(\omega)}X_s (\omega)\]
		as $\sup_{0\leq s<t_k(\omega)}X_s(\omega)\geq \sup_{0\leq s<\tau^{\frac1n}(\omega)}X_s (\omega)$.
		Letting $k\to\infty$ this implies by right continuity of $(X_t)_t$
		\begin{equation}
			\label{ineqstopping}
			\frac{1}{n}\leq X_{\tau^{\frac1n}(\omega)}(\omega)-\sup_{0\leq s<\tau^{\frac1n}(\omega)}X_s (\omega).
		\end{equation}

		Take now $s_k:=(Z-\frac{1}{k}\ind_B)^+$ for $k\in\mathbb{N}$. $s_k$ is $\Fcal-\Borel_{[0,+\infty)}$ measurable and bounded. Observe that since $(X_t)_t$ is RCLL, by \cite{KS91} Remark 1.1.14 we have $(t,\omega)\mapsto X_t(\omega)$ is $\Borel_{[0,+\infty)}\otimes \Fcal-\Borel_{[0,+\infty)}$ measurable, and $\omega\mapsto (\omega, Y(\omega))$ is $\Fcal- \Fcal\otimes\Borel_{[0,+\infty)}$  measurable for every $\Fcal-\Borel_{[0,+\infty)}$ measurable bounded $Y:\Omega\rightarrow [0,+\infty)$. Hence $\omega\mapsto X_{s_k(\omega)}(\omega)$ is $\Fcal-\Borel_{[0,+\infty)}$ measurable.  By continuity from below, argued at the beginning of this proof, 
		\begin{equation}
			\label{limitsuk}T_u\big((s_k-M)\ind_B\big)\rightarrow_kT_u\big((Z-M)\ind_B\big).
		\end{equation}
		At the same time, $s_k(\omega)=Z(\omega)=0$ for every $\omega\in\Omega\setminus B$ and using \eqref{ineqstopping}  $X_{s_k(\omega)}(\omega)\leq \sup_{0\leq s<\tau^{\frac1n}(\omega)}X_s (\omega)\leq -\frac{1}{n}+ X_{\tau^{\frac1n}(\omega)}(\omega)$, for any $\omega\in B$. Recalling that $u(\omega,0)=0$ for every $\omega\in\Omega$, we can then write 
		\begin{align*}
			T_u\big((s_k-M)\ind_B\big)&= \int_B u(\omega, s_k(\omega)-M)\dP= \int_B X_{s_k(\omega)}(\omega)\dP \\
			&\stackrel{\eqref{ineqstopping}}{\leq} -\frac{1}{n}\PWp(B)+\int_B X_{\tau^{\frac1n}(\omega)}(\omega)\dP
			\\& =-\frac{1}{n} \PW(B)+\int_\Omega u(\omega, Z(\omega)-M)\dP
		\end{align*}
		so that $T_u((s_k-M)\ind_B)+\frac{1}{n}\PW(B)\leq T_u((Z-M)\ind_B)$ reaching a contradiction with \eqref{limitsuk}, as $\frac{1}{n}\PW(B)>0$ does not depend on $k$. Now, we conclude that $\PWp(A^M_{cont})=1$ for every arbitrarily fixed $M>0$. Since $A_{cont}=\bigcup_{M\in\mathbb{N}}A^M_{cont}$ we conclude that both $A_{cont}\in\Fcal$ and $\PWp(A)=1$.

	\end{proof}

	\begin{corollary}
		\label{corexistscontversion}
		Under the same assumptions of Proposition \ref{propcontofuplus}, there exists a function $\widehat{u}$ such that $\widehat{u}(\omega,\cdot)$ is nondecreasing and continuous for every $\omega\in\Omega$, $E=\{\omega\in\Omega\mid u(\omega, x)=\widehat{u}(\omega,x)\,\forall x\in\R\}\in\Fcal$ and $\PWp(E)=1$.  In particular, $T_u(f)=T_{\widehat{u}}(f),\,\forall f\in \mathcal{L}^\infty(\Omega,\Fcal)$. Finally, if $u(\omega,\cdot)$ is strictly increasing on $\mathbb{R}$ for every $\omega\in\Omega$, $\widehat{u}(\omega,\cdot)$ can be taken strictly increasing on $\mathbb{R}$ for every $\omega\in\Omega$.
	\end{corollary}
	\begin{proof}
		As a consequence of Proposition \ref{propcontofuplus}, $A_{cont}\in\Fcal$. Set $A:=A_{cont}$ and define now $\widehat{u}(\omega,x):=u(\omega,x)\ind_A(\omega)+x\ind_{\Omega\setminus A}$. Then clearly $\widehat{u}(\omega,\cdot)$ is continuous and nondecreasing either $\omega\in A$ or $\omega\in \Omega\setminus A$. Further for any fixed $x\in\R$ the map $\omega\mapsto \widehat{u}(\omega,x)$ is $\Fcal$ measurable. As $u(\omega, x) \in \Lcal^1(\Omega, \Fcal, \PW)$ implies $\widehat{u}(\omega, x) \in \Lcal^1(\Omega, \Fcal, \PW)$ for every $x \in \R$ we can invoke Lemma \ref{lemma:wellposed} to obtain that $\widehat{u}$ is $\Fcal \otimes \Borel_\R - \Borel_\R$ measurable, $T_{\widehat{u}}$ is well defined in $\Lcal^{\infty}(\Omega, \Fcal)$ and $T_u \equiv T_{\widehat{u}}$.
		
	\end{proof}

	\bibliographystyle{abbrv}

\end{document}